\documentclass[11pt,a4paper,reqno]{amsart}%
\usepackage{amsthm,amsmath,amsfonts,amssymb,amsxtra,appendix,bookmark,dsfont,latexsym}
\usepackage{amsfonts,mathrsfs}

\makeatletter
\usepackage{hyperref}
\usepackage{delarray,a4,color}
\usepackage{euscript}

\usepackage{bbm}
\usepackage[latin1]{inputenc}

\theoremstyle{plain}
\newtheorem{theorem}{Theorem}[section]
\newtheorem{lemma}[theorem]{Lemma}

\newtheorem{proposition}[theorem]{Proposition}

\theoremstyle{remark}
\newtheorem{remark}[theorem]{Remark}

\numberwithin{equation}{section}

\DeclareMathOperator{\spec}{spec}

\DeclareMathOperator{\Tr}{Tr}

\def\geqslant{\ge}
\def\leqslant{\le}
\def\bq{\begin{eqnarray}}
	\def\eq{\end{eqnarray}}
\def\bqq{\begin{eqnarray*}}
	\def\eqq{\end{eqnarray*}}
\def\nn{\nonumber}

\newcommand{\norm}[1]{\left\lVert #1 \right\rVert}

\newcommand\1{{\ensuremath {\mathds 1} }}

\newcommand{\im}{\mathrm{i}}

\renewcommand{\epsilon}{\varepsilon}
\newcommand{\bx}{\mathbf{x}}

\newcommand{\by}{\mathbf{y}}

\def\R {\mathbb{R}}

\def\R {\mathbb{R}}

\def\d{{\rm d}}

\def\b{|}

\renewcommand{\leq}{\leqslant}
\renewcommand{\geq}{\geqslant}

\newcommand{\bA}{\mathbf{A}}

\newcommand{\bz}{\mathbf{z}}

\newcommand{\curl}{\mathrm{curl}}

\title{A Lieb-Thirring inequality for extended anyons}

\author[T.Girardot]{Th\'{e}otime Girardot}
\address{Aarhus university, Nordre Ringgade 1, 8000 Aarhus C}
\email{theotime.girardot@math.au.dk}
\author[N.Rougerie]{Nicolas Rougerie}
\address{Ecole Normale Sup\'erieure de Lyon \& CNRS, UMPA (UMR 5669)}
\email{nicolas.rougerie@ens-lyon.fr}

\date{September 2022}

\begin{document}

\begin{abstract}
We derive a Pauli exclusion principle for extended fermion-based anyons of any positive radius and any non-trivial statistics parameter. That is, we consider 2D fermionic particles coupled to magnetic flux tubes of non-zero radius, and prove a Lieb-Thirring inequality for the associated many-body kinetic energy operator. The corresponding constant is independent of the radius of the flux tubes, and proportional to the statistics parameter. 
\end{abstract}

\maketitle

\setcounter{tocdepth}{2}

 
\makeatletter
\def\@tocline#1#2#3#4#5#6#7{\relax
  \ifnum #1>\c@tocdepth 
  \else
    \par \addpenalty\@secpenalty\addvspace{#2}%
    \begingroup \hyphenpenalty\@M
    \@ifempty{#4}{%
      \@tempdima\csname r@tocindent\number#1\endcsname\relax
    }{%
      \@tempdima#4\relax
    }%
    \parindent\z@ \leftskip#3\relax \advance\leftskip\@tempdima\relax
    \rightskip\@pnumwidth plus4em \parfillskip-\@pnumwidth
    #5\leavevmode\hskip-\@tempdima
      \ifcase #1
       \or\or \hskip 1em \or \hskip 2em \else \hskip 3em \fi%
      #6\nobreak\relax
      \dotfill
      \hbox to\@pnumwidth{\@tocpagenum{#7}}
    \par
    \nobreak
    \endgroup
  \fi}
\makeatother
\tableofcontents

\section{Introduction}\label{sec:intro}

The exclusion principle of quantum physics can be formulated in terms of Lieb-Thirring inequalities for the kinetic energy of fermionic particles. These inequalities stay true in various contexts, for instance, when the particles feel an external magnetic field. In this paper we establish a Lieb-Thirring inequality for extended anyons, modeled as fermions coupled to magnetic flux tubes of finite radius. The magnitude $\alpha\in [0,2]$ of the magnetic flux is interpreted as $1 + $ the statistics parameter (because our basic wave-functions are fermionic). Our motivation is two-fold: 
\begin{enumerate}
 \item We bridge a gap in the current state of the research program initiated in~\cite{LunSol-13a}. Indeed, in~\cite{LunSei-18} a Lieb-Thirring inequality is established for any $\alpha \neq 1$, at zero radius. Thus, ideal anyons of any statistics except the bosonic one satisfy a Pauli exclusion principle. On the other hand, in~\cite{LarLun-16}, bounds suggestive of a Pauli principle are proven for finite radii $R$, but under restrictive assumptions on $\alpha$ and $R$.
 \item We improve the main results~\cite[Theorems 1.1 and 1.4]{GirRou-21} of a previous paper of ours. Indeed, directly conditioned on the inequality we prove below, a semi-classical effective model for almost-fermionic extended anyons could be derived under relaxed assumptions in a mean-field type limit. This was mentioned in~\cite[Remark~1.2]{GirRou-21} and proved\footnote{The final improvement is actually from $\eta < 1/4$ to $\eta < 1/3$ in the notation of these references, slightly less than we had hoped for in~\cite[Remark~1.2]{GirRou-21}.} in the first author's phd thesis~\cite[Chapter~15]{Girardot-21}.
\end{enumerate}

Before stating our inequality precisely we quickly recall basic facts about the two main concepts of the paper: extended anyons and Lieb-Thirring inequalities. In particular we give precisions on the vocabulary used above. We will not consider the case of non-abelian anyons where the exchange phase is replaced by a general unitary operator. See~\cite{LunQva-20} and references therein.

\subsection{Anyons}
In a quantum mechanical system, dimensionality plays a fundamental role. In three or higher dimensions, the indistinguishability of the particles naturally leads to sorting them out into two types, bosons and fermions. These two types have different statistical behavior, leading to the commonly used terminology of bosonic (or fermionic) \emph{statistics}. This dichotomy no longer holds in two dimensions. The richer topology of the configuration space indeed allows for more than two statistics. Consider a wave function $\Psi :\left (\R^{2}\right )^{N}\to \mathbb{C}$,
it will formally behave as
\begin{equation}\label{eq:stat}
	\Psi (\bx_{1},...,\bx_{j},...,\bx_{k},...,\bx_{N})=e^{\im\alpha_b \pi}\Psi(\bx_{1},...,\bx_{k},...,\bx_{j},...,\bx_{N})
\end{equation}
where $\alpha_b\in \left [-1,1\right ]$ is the statistics parameter of the anyons, counted from the bosonic end. The case $\alpha_b =0$ corresponds to bosons and $\alpha_b =1$ to fermions. The possibility of different statistics have been known since the 70's from different approaches~\cite{LeiMyr-77,Myrheim-99,Wilczek-82a,GolMedSha-81} and used to describe quasi-particles emerging in the fractional quantum Hall effect ~\cite{Halperin-84,ChenWilWitHal-89,GreiWenWil-91,AroSchWil-84,ChenWilWitHal-89,GolSha-96,LunRou-16}, rotating Bose gases~\cite{ZhaSreJai-15,CooSim-15} and in quantum information~ \cite{NayCheSimSteSter-08}. There are two main ways to model anyons. We can treat $\Psi$ either as a multi-valued function (a section of a complex line bundle) or, as a usual bosonic or fermionic function with a modified kinetic energy. We follow the latter approach, called the magnetic gauge picture. 

We thus consider a fermionic wave function whose kinetic energy is modified through a singular change of gauge. Namely, we encode the behaviour of the wave function \eqref{eq:stat} under a particle exchange by setting
$$\Psi(\bx_{1},...,\bx_{N})=\prod_{j<k}e^{\im\alpha_f \phi_{jk}}\Phi(\bx_{1},...,\bx_{N}) \;\:\text{where}\;\:
\phi_{jk}=\mathrm{arg}\frac{\bx_{j}-\bx_{k}}{|\bx_{j}-\bx_{k}|}$$
with $\Phi$ a fermionic wave function, antisymmetric under particle exchange and 
\begin{equation}\label{eq:bosfer}
 \alpha_f = 1 + \alpha_b 
\end{equation}
the statistics parameter, now counted from the fermionic end.  We have denoted $\mathrm{arg}(\,.\,)$ the angle of a planar vector with the horizontal axis. The case $\alpha_f=0$ describes a usual fermionic system.
Applying this transformation, the momentum operator for particle $j$ changes as
\begin{equation}\label{eq:kin op}
-\im\nabla_{\bx_{j}} \to D_{j}:=-\im\nabla_{\bx_{j}} +\alpha_f \bA\left (\bx_{j}\right ) 
\end{equation}
with
\begin{equation}
	\bA\left (\bx_{j}\right ):=\sum_{k\neq j}\frac{\left (\bx_{j}-\bx_{k}\right )^{\perp}}{\left \b \bx_{j}-\bx_{k}\right \b^{2}}
	\label{AJP}
\end{equation}
where $\left (x,y\right )^{\perp}=\left (-y,x\right )$. Namely, we have the formal identity 
$$ \left\langle \Psi | \left(-\im\nabla_{\bx_{j}} \right) ^2 \Psi \right\rangle = \left\langle \Phi | D_{\bx_j} ^2 \Phi \right\rangle.$$
The above is a description of ideal anyons: the system behaves like ordinary particles  attached to infinitely thin solenoids perpendicular to the plane. In other words we added a magnetic Aharonov-Bohm type interaction between the particles, which can formally be gauged away by changing the symmetry type of wave-functions. This corresponds to the particular case of anyons of radius $R=0$. In this work we are interested in deriving Lieb-Thirring inequalities for anyons of radius $R>0$, meaning that the Aharonov-Bohm flux tube will be smeared over a finite radius. Before we describe this in details we recall what is known for the above model at $R=0$. We refer to~\cite{AdaTet-98,LunSol-13a,LunSol-13b,CorOdd-18,CorFer-21} for more details on the definition of the model, in particular different possible self-adjoint extensions. In the sequel we always use the Friedrichs extension.

\subsection{Lieb-Thirring inequalities for ideal anyons}
The celebrated Lieb-Thirring inequalities are one of the different ways to quantify the Pauli exclusion between fermionic particles. We know, for instance \cite[Theorem 4.3]{LieSei-09} that, in two dimensions, fermions exclude one another in the sense that for any anti-symmetric wave-function $\Psi$ normalized in $L^{2}\left (\R^{2N}\right )$ 
\begin{equation}
	\sum_{j=1}^{N}\int_{\R^{2N}}\left \b \nabla_{\bx_j}\Psi\right \b^{2}\d \bx \geq C_{2}^{\mathcal{K}}\int_{\R^{2}}\rho\left (\bx\right )^{2}\d \bx
	\label{LTF}
\end{equation}
where
\begin{equation}
	\rho(\bx):=\sum_{j=1}^{N}\int_{\R^{d(N-1)}}\left \b \Psi(\bx_{1},...,\bx_{j}=\bx,...,\bx_{N})\right \b^{2}\prod_{k\neq j}\d \bx_{k}\nn
\end{equation}
is the one-particle density. This inequality remains true with $-\im\nabla\to -\im\nabla+\bA$ for some suitable magnetic vector potential $\bA$. It implies that the kinetic energy of a large number $N$ of fermionic particles inside a bounded domain $\Omega \subset \R^2$ must grow proportionally to $N^2 \gg N$: fermions do not like to be together in the same quantum state and make it known through an energetic cost growing with the number of particles.

On the other hand, the best we can achieve with bosons is
\begin{equation}
	\sum_{j=1}^{N}\int_{\R^{dN}}\left \b\nabla_{\bx_j}\Psi\right \b^{2}\d \bx\geqslant \frac{C_{2}}{N}\int_{\R^{d}}\rho(\bx)^{2}\d \bx
	\label{LTB}
\end{equation}
which is actually the Sobolev inequality, and becomes trivial as $N\to \infty$. Hence the kinetic energy of a large number $N$ of bosonic particles inside a bounded domain $\Omega \subset \R^2$ may well stay of order $N$.

As regards anyons, we intuitively think that the closer they are to fermions, the more they exclude one another. This directly leads to the idea of a Lieb-Thirring inequality for anyons, proportional to $\alpha$. Results in this direction are available in \cite{LunSol-13a,LunSol-13b,LunSol-14}. Here we quote the more recent~\cite{LunSei-18} where a bound is obtained for any $\alpha_{b}\in\left [-1,1\right ]$:

\begin{theorem}[\textbf{Lieb-Thirring inequality for ideal anyons}]\label{Ltida}\mbox{}\\
	Let $D_{\bx_j}$ be as in~\eqref{eq:kin op}. There exists a constant $C>0$ such that for any $\alpha_f \in [0,2]$, $N\geqslant 1$ and $\Psi_{N}\in L^{2}_{\mathrm{asym}}\left (\R^{2N}\right )$ with $\norm{\Psi_{N}}_{L^2}=1$ we have
	\begin{equation}
		\sum_{j=1}^{N}\int_{\R^{dN}}\left \b D_{\bx_j} \Psi_{N}\right \b^{2}\d x\geqslant C\left|1-\alpha_f\right|\int_{\R^{d}}\rho_{\Psi_{N}}(\bx)^{2}\d \bx .\nn
	\end{equation}
\end{theorem}

Beware that we use the fermionic convention that wave-functions are anti-symmetric. In~\cite{LunSei-18} the result is stated with $\Psi_{N}\in L^{2}_{\mathrm{sym}}\left (\R^{2N}\right )$ and $1-\alpha_f$ replaced by $\alpha_b$ as per~\eqref{eq:bosfer}. This inequality shows that anyonic particles of any type but the bosonic one ($\alpha_b = 0, \alpha_f = 1$) satisfy a Lieb-Thirring inequality and thus an exclusion principle.

The usual approach to proving Lieb-Thirring inequalities for fermions is to see them as dual to bounds on eigenvalue sums for Schr\"odinger operators, and use the Birman-Schwinger principle, see~\cite[Chapter~4]{LieSei-09} for review. This clearly does not apply in the anyonic case, because the problem is genuinely many-body. A new approach based on the local exclusion came up in the past ten years~\cite{FraSei-12,LunSol-14,LunSol-13a,LunPorSol-15,LunSei-18,LunNamPor-16,LarLunNam-21,Nam-22,Nam-20}. It consists in proving possibly $N$-dependent inequalities on finite subsets of $\R^{d}$. A clever covering of the space then allows to patch the inequalities together and obtain the correct $N$-dependence on the whole space. We will employ this technique here, applying it to the Hamiltonian for extended anyons that we describe next (see also~\cite{LarLun-16,LunRou-15,Girardot-19,GirRou-21}).

\subsection{Model for extended anyons}

Consider the 2D Coulomb potential generated by a unit charge smeared over the disk of radius $R$
\begin{equation}
	w_{R}(\bx)=\left(\log\b \;.\;\b *\chi_{R}\right)(\bx),\:\:\text{with the convention}\:\:w_{0}=\log\b \;.\;\b
	\label{wrr}
\end{equation}
and $\chi_{R}(x)$ a positive, regularizing function of unit mass
\begin{equation}
	\chi_{R}\left (\bx\right ):=\frac{\1_{B\left (0,R\right )}\left (\bx \right )}{\pi R^{2}}.
	\label{khiRr}
\end{equation}
Observe that
\begin{align*}
	\nabla^{\perp}w_{0}(\bx)=\frac{\bx^{\perp}}{|\bx|^{2}},\:\:\text{and}\:\:B_{0}(\bx):=\nabla^{\perp}\nabla^{\perp}w_{0}=\Delta w_{0}=2\pi\delta_{0}
\end{align*}
so that we recover the magnetic field of the ideal anyon case (in a distributional sense) at $R= 0$. A natural regularisation of the ideal anyons potential vector~\eqref{AJP} is
\begin{equation}
	\bA^{R}(\bx_{j}):=\sum_{k\neq j}\nabla^{\perp}w_{R}(\bx_{j}-\bx_{k}):=\sum_{k\neq j}\frac{\left (\bx_{j}-\bx_{k}\right )^{\perp}}{\left \b \bx_{j}-\bx_{k}\right \b_{R}^{2}}.
	\label{AJRr}
\end{equation}
where we have introduced the regularized distance
$$\left \b \bx\right \b_{R}:=\max \left \{\left \b\bx\right \b ,R\right \}.$$
The magnetic field felt by particle $j$ is then
\begin{equation}
	\curl \bA^{R}(\bx_{j})=2\pi\sum_{k\neq j}\frac{\1_{B\left (\bx_{k},R\right )}\left (\bx_{j} \right )}{\pi R^{2}},\nn
\end{equation}
i.e. it sees all the other particles as carrying a tube of flux of radius $R$. From now on we only use the fermionic representation and set 
$$ \alpha := \alpha_f = 1 + \alpha_b.$$
We always assume that 
$$ \alpha \in [0,2],$$
which, unlike in the ideal anyon model, is a true restriction, for one cannot restrict to this case by a change of gauge. We could consider $\alpha \notin [0,2]$ but in this case our bounds would depend on $\alpha$ mod 2, which we believe is optimal. This is certainly the case for the bounds of Section~\ref{sec:large} below, where we deal with a parameter regime where the smearing of flux-tubes is shown to be negligible.

The full kinetic energy operator is 
\begin{equation}
	T_{\alpha}^{R}:=\sum_{j=1}^{N}\left (D_{\bx_j}^{R}\right )^{2}:=\sum_{j=1}^{N}\left (-\im\nabla_{\bx_j}+\alpha\bA^{R}(\bx_{j})\right )^{2}
	\label{energy_kin_R}
\end{equation}
acting on the fermionic space $L^{2}_{\rm asym}\left (\R^{2N}\right )$ as an unbounded operator. We denote $\mathcal{D}_{\alpha ,R}^{N}$ the domain of \eqref{energy_kin_R}. When $R>0$, $\bA^{R}$ is a bounded perturbation of $-\im\nabla$. The kinetic energy $T^{R}_{\alpha}$ is essentially self-adjoint on its natural domain (see \cite[Theorem X.17]{ReeSim2} and~\cite{AvrHerSim-78}). The bottom of its spectrum exists for any fixed $R>0$.

Apart from providing an analytically useful regularisation of the model, the above Hamiltonian with smeared flux tubes is actually the relevant one for emergent anyons~\cite{LunRou-16,LamLunRou-22} in the fractional quantum Hall effect. The size of the flux tubes is set by the magnetic length of the host system. Early considerations of the model are in~\cite{Trugenberger-92,Trugenberger-92b}.

\subsection{Main theorem}

We define the kinetic energy of the $N$-particles system with wave function $\Psi_{N}$ as
\begin{equation}\label{KE}
	\mathcal{E}_{\alpha}^{R}\left [\Psi_{N}\right ]=\sum_{j=1}^{N}\int_{\R^{2N}}\left \b \left (-\im\nabla_{\bx_j}+\alpha\bA^{R}\left (\bx_{j}\right )\right )\Psi_{N}\right \b^{2}\d \bx_{1} \ldots \d \bx_{N}.
\end{equation}
We also denote the one-body density of the system 
\begin{equation}
	\rho_{\Psi_{N}}(\bx):=\sum_{j=1}^{N}\int_{\R^{d(N-1)}}\left \b \Psi_{N}(\bx_{1},...,\bx_{j}=\bx,...,\bx_{N})\right \b^{2}\prod_{k\neq j}\d \bx_{k}\nn
\end{equation}
and state the main theorem of the paper.

\begin{theorem}[\textbf{Lieb-Thirring inequality for extended anyons}]\mbox{}\label{LT_ex_an}\\
	There exists a constant $C^{\mathrm{EA}}$ independent of $R$, $\alpha$ and $N$ such that for any $L^2$-normalized fermionic $N$-particles wave function $\Psi_{N}\in L^{2}_{asym}\left (\R^{2N}\right )$ 
	\begin{equation}
		\mathcal{E}_{\alpha}^{R}\left [\Psi_{N}\right ]\geq C^{\mathrm{EA}}\left| \alpha -1\right |\int_{\R^{2}}	\rho^{2}_{\Psi_{N}}(\bx)\d \bx.\nn
	\end{equation}
\end{theorem}

\subsection{Strategy of proof}

The rest of the paper is dedicated to the proof of Theorem \ref{LT_ex_an}. We follow the route opened in \cite{LunSol-13a,LunSol-13b,LunSol-14} where, roughly speaking, the authors divide the plane into squares $Q$, of various sizes $\left \b Q\right \b=L^{2}$, on which they prove a local Lieb-Thirring inequality before recombining. To do so, one needs
\begin{enumerate}
	\item \textbf{A local exclusion principle} stating that, on a given square, the presence of two anyons\footnote{Of course, the property ``a square contains a certain number of particles'' is a probabilistic statement.} is sufficient to get a lower bound on the kinetic energy proportional to $L^{-2}$. For usual fermions this just means that the kinetic energy can have only one zero mode per box (namely, the constant function), so that only one particle per box can have zero kinetic energy. 
	\item \textbf{A local uncertainty principle} derived from Poincar\'e-Sobolev inequalities. This does not use the statistics and is equally valid for bosons. Combining with the local exclusion yields a local Lieb-Thirring inequality on squares with more than two particles. 
	\item \textbf{A smart splitting algorithm} ensuring that the total energy of the squares with more than two particles is sufficient to compensate for squares on which we do not have enough particles to obtain the inequality. 
\end{enumerate}

In the sequel we essentially keep the same framework with the main following steps:
\begin{enumerate}
	\item \textbf{A local exclusion principle for extended anyons} stating that if a box contains more than a fixed number $\underline{N}$ of particles, its energy must be positive, proportional to $L^{-2}$. The proof technique we use rather depends on the ratio $\gamma =R/L$. 
	\item \textbf{A local uncertainty principle} derived from the diamagnetic and Sobolev inequalities and quite similar to the previously mentioned one.
	\item \textbf{The Besicovitch covering theorem} allowing us to cover the plane with sets each containing sufficiently many particles to apply the local exclusion principle, while intersecting one another only a finite number of times.
\end{enumerate}
The third idea was introduced in the recent article~\cite{Nam-22}, and allows to think purely locally, without having to look for compensations\footnote{We could have used the splitting algorithms of~\cite{LunSol-13a,LunNamPor-16} in our proof instead of the Besicovitch theorem.} between  different spatial regions as in~\cite{LunSol-13a,LunSol-13b,LunSol-14}. Hence  our main task is to provide the local exclusion principle. 

If $\gamma = R /L \ll 1$ (large boxes), the fact that the anyons are extended with $R>0$ intuitively does not play a very big role, and we can adapt arguments from~\cite{LunSei-18,LarLun-16} to obtain the exclusion principle. The difference is that in~\cite{LunSei-18} the influence of particles outside of the box can be gauged away freely because the attached magnetic flux is purely local. We prove that for $\gamma \ll 1$ this influence can be gauged away at a small, controlable cost. 

If $\gamma = R /L \sim 1$ (medium boxes) we can use the well-known inequality (combine~\cite[Theorem~7.21]{LieLos-01} and~\cite[Lemma~1.4.1]{FouHel-book})
\begin{equation}\label{eq:basic mag}
 \left\langle \psi , (-\im \nabla + A) ^2 \psi \right\rangle \geq \frac{1}{2} \left\langle |\psi |, (-\Delta + \curl A) |\psi| \right\rangle 
\end{equation}
to obtain bounds using a two-body model (the $\curl$ of~\eqref{AJRr} is a pair interaction). For $\gamma \sim 1$ the Dyson lemma~\cite{LieSeiSolYng-05,Rougerie-EMS} allows to use the kinetic energy to smear the two-body interaction over the whole box and get a non-trivial lower bound. This argument is worked out in~\cite{LarLun-16}, whose results we quote and adapt to our situation. The fermionic symmetry cannot be used efficiently with this method (because $|\psi|$ and not $\psi$ itself appears in the right side of~\eqref{eq:basic mag}). Consequently the dependence on $\alpha$ of the bound would degenerate around $\alpha=0$ (fermionic end). We remedy this by treating the magnetic field perturbatively in this regime, relying on the bounds for free fermions. 

If $\gamma = R /L \gg 1$ (small boxes), the box is completely covered by the magnetic flux attached to each particle inside it. We can then show that the problem becomes effectively one-body. Intuitively, the magnetic field does not harm the fact that the kinetic energy only has one zero-mode on the box. We prove a diamagnetic bound vindicating that there are only finitely many modes with energy less than $L^{-2}$, uniformly in the magnetic field. Since our wave-functions are fermionic, this implies the exclusion principle lower bound if sufficiently many particles are in the box.

Strictly speaking~\eqref{eq:basic mag} is not available on a box with Neumann boundary conditions. Hence, throughout the paper we apply it first on the whole space to obtain
\begin{align}\label{KEbis}
	\mathcal{E}_{\alpha}^{R}\left [\Psi_{N}\right ] &\geq \int_{\R^{2N}} \frac{1}{2} \sum_{j=1}^{N}\left \b \left (-\im\nabla_{\bx_j}+\alpha\bA^{R}\left (\bx_{j}\right )\right )\Psi_{N}\right \b^{2}\d \bx_{1} \ldots \d \bx_{N}\nonumber\\
	&+ \int_{\R^{2N}} \frac{1}{4} \left(\sum_{j=1}^{N}\left \b \nabla_{\bx_j} |\Psi_N|\right \b^{2} + 2\pi\alpha\sum_{j=1}^{N}\sum_{k\neq j}\frac{\1_{B\left (\bx_{k},R\right )}\left (\bx_{j} \right )}{\pi R^{2}} |\Psi_N|^2\right)\d \bx_{1} \ldots \d \bx_{N}\nonumber \\
	&=: \int_{\R^{2N}} \left(e_1 (\Psi_N) + e_2 (\Psi_N) \right)\d \bx_{1} \ldots \d \bx_{N}
\end{align}
and derive lower bounds on $	\mathcal{E}_{\alpha}^{R}$ using different terms of the right-hand side of the above for different ranges of the parameter $\gamma $.

\begin{remark}[Fermion-based anyons and diamagnetic bounds]\label{rem:fermions}\mbox{}\\
	The fermionic nature of the wave-functions we work with plays a crucial role in the case of small and medium boxes (in the latter case, only for small $\alpha$). In medium boxes, for most values of $\alpha$ we mostly use an already existing result from~\cite{LarLun-16} which is independent of the statistics of the basic wave-functions. In the case of large boxes we improve the results of~\cite{LarLun-16} by combining them to the techniques of~\cite{LunSei-18}. We then replace a bound proportional to $\alpha_{N}$ with the complicated behavior \eqref{eq:alphaN} by a bound proportional to $\alpha$.
	
	When we do use it, the fermionic symmetry of wave-functions enters after reducing the desired bounds to a one-body problem. We then use diamagnetic estimates to obtain bounds independent of the remaining  magnetic field (which can be quite general). We are indebted to the anonymous referees of the paper for pointing out that the bound we use (Lemma~\ref{lem:N}) can be obtained by combining results from~\cite{Frank-09,HunSim-04}. We nevertheless provide our proof in Appendix~\ref{sec:app} for the convenience of the reader. \hfill$\diamond$
\end{remark}
\bigskip

In Section~\ref{sec:key} we state our local exclusion principle bound and explain how to deduce~Theorem \ref{LT_ex_an} using local uncertainty and the Besicovitch theorem, as in~\cite{Nam-22}. The heart of the paper is then Section~\ref{sec:local} where we prove the local exclusion estimate, distinguishing according to the size of the box. All in all, the logical structure of the argument is 
\begin{equation}\label{eq:structure}
\mbox{Section~\ref{sec:local}} \Rightarrow \mbox{Theorem~\ref{LEP_an}} \Rightarrow \mbox{Theorem~\ref{th:LT_FN}} \Rightarrow \mbox{Theorem~\ref{LT_ex_an}}.
\end{equation}
We dispose of the last two implications first, with essentially known methods, in order to focus on the key new estimates in Section~\ref{sec:local}.

\bigskip

\noindent\textbf{Acknowledgments.} 
We thank Douglas Lundholm for insightful discussions. Funding from the European Research Council (ERC) under the European Union's Horizon 2020 Research and Innovation Programme (Grant agreement CORFRONMAT No 758620) is gratefully acknowledged as well as the grant 0135-00166B from Independent Research Fund Denmark.

\section{Reduction to local estimates with finite $N$}\label{sec:key}

As usual with Lieb-Thirring inequalities, the main point of Theorem~\ref{LT_ex_an} is the optimal dependence on $N$. In this section we reduce the proof to local estimates with unspecified $N$ dependence via the Besicovitch covering theorem, following~\cite{Nam-22}.

\subsection{Reduction to local Lieb-Thirring at finite $N$}

Let $Q\subset \R^2$ a bounded domain (always taken to be a square in the sequel). We denote the local kinetic energy in $Q$ by
\begin{equation}\label{def:EQ}
	\mathcal{E}_{Q}^R\left [ \Psi_{N}\right ]:=\sum_{j=1}^{N}\int_{\R^{2N}}\left(e_1 (\Psi_N) + e_2 (\Psi_N) \right)\1_{Q}\left (\bx_{j}\right ) \d x_{N}
\end{equation}
with the energy densities $e_1,e_2$ as defined in~\eqref{KEbis}.
We drop the $\alpha$ dependence from the notation. All quantities implicitly depend on $\alpha$ unless we explicitly state otherwise. Our local Lieb-Thirring inequality, to be derived in Section~\ref{sec:local LT} is as follows

\begin{theorem}[\textbf{Lieb-Thirring at finite $N$}]\mbox{}\label{th:LT_FN}\\
	There exist three numbers $N_<$, $N_>$ and $C^{\mathrm{FN}}$ independent of $\alpha $ and $R$ such that if $\Psi_{N}\in L^{2}_{\mathrm{asym}}\left (\R^{2N}\right )$ is a $L^2$-normalized wave-function and $Q$ a square for which
	\begin{equation}\label{as:LT_FN}
		N_< \leq \int_{Q}\rho_{\Psi_{N}}\left (\bx\right )\d \bx\leq N_>
	\end{equation}
	 we have that
	\begin{equation}
		\mathcal{E}_{Q}^R \left [ \Psi_{N}\right ]\geq C^{\mathrm{FN}}\left |\alpha -1\right |\int_{Q}	\rho^{2}_{\Psi_{N}}(\bx)\d \bx.\nn
	\end{equation}
\end{theorem}

We now explain how this implies Theorem~\ref{LT_ex_an}. We recall the Besicovitch covering lemma~\cite{Besicovitch-45,Besicovitch-46} which was used to prove Lieb-Thirring inequalities in~\cite{Nam-22} (using balls instead of squares).

\begin{lemma}[\textbf{Besicovitch covering lemma}]\mbox{}\label{th:BES}\\
	Let $E$ be a bounded subset of $\R^{d}$. Let $\mathcal{F}$ be a collection of hypercubes in $\R^{d}$ with faces parallel to the coordinate planes such that every point $\bx\in E$ is the center of a cube from $\mathcal{F}$. Then there exists a sub-collection $\mathcal{G}\subset \mathcal{F}$ such that
	\begin{equation}
		\1_{E}\leq\sum_{Q\in \mathcal{G}}\1_{Q}\leq b_{d} \1_{E},
	\end{equation}
	namely, $E$ is covered by $\bigcup_{Q\in \mathcal{G}}Q$ and every point in $E$ belongs to at most $b_{d}$ cubes from $\mathcal{G}$. The constant $b_{d}$ only depends on the dimension $d\geq 1$.
\end{lemma}

\begin{proof}[Proof of Theorem \ref{LT_ex_an}, last implication in~\eqref{eq:structure}]\mbox{}\label{pr:Main}\\
	We start by considering $\Psi_{N}$'s for which the total number of particles is bounded from above by $N_{<}$
	\begin{equation}
		\int_{\R^{2}}\rho_{\Psi_{N}}\left (\bx\right )\d \bx < N_< 
	\end{equation}
 Then Theorem \ref{th:LT_FN} does not apply because of the assumption \eqref{as:LT_FN}. 
	In this case, we can however apply the diamagnetic~\cite[Theorem 7.21]{LieLos-01} and Sobolev inequalities
	\begin{equation}
		\mathcal{E}_{\alpha}^{R}\left [\Psi_{N}\right ]\geq \sum_{j=1}^{N}\int_{\R^{2N}}\left \b -\im\nabla_{\bx_j}\left \b \Psi_{N}\right \b\right \b^{2}\d x_{N}\geq\frac{1}{C_{2}} \frac{\int_{\R^{2}}\rho^{2}_{\Psi_{N}}\left (\bx\right )\d \bx}{\int_{\R^{2}}\rho_{\Psi_{N}}\left (\bx\right )\d \bx}\geq \frac{1}{C_{2}N_<}\int_{\R^{2}}\rho^{2}_{\Psi_{N}}\left (\bx\right )\d \bx.\nn
	\end{equation}
	There remains to consider the case
	\begin{equation}
		\int_{\R^{2}}\rho_{\Psi_{N}}\left (\bx\right )\d \bx \geq N_< 
	\end{equation}
	By a density argument we may assume that $\Psi_N$ is smooth, with compact support in $E^N$ with $E\subset \R^2$ bounded. Then $\rho_{\Psi_{N}}$ is continuous with a bounded support $E\subset \R^{2}$, and we can for every $\bx\in E$, find a square $Q_{\bx}\subset \R^{2}$ centered at $\bx$ such that
	\begin{equation}\label{eq:enro}
		\int_{Q_{\bx}}\rho_{\Psi_{N}}\left (\by\right )\d \by =\frac{N_< + N_>}{2}
	\end{equation}
	where $N_>$ is as in the statement of Theorem~\ref{th:LT_FN}. We apply the Besicovitch covering Lemma~\ref{th:BES} to the collection of squares $\mathcal{F}=\left \{Q_{\bx}\right \}_{\bx \in E}$ to obtain a sub-collection $\mathcal{G}\subset\mathcal{F} $ such that
	\begin{equation}\label{ine:bes}
		\1_{E}\leq\sum_{Q\in \mathcal{G}}\1_{Q}\leq b_{2} \1_{E}.
	\end{equation}
	The second inequality above implies that
	\begin{equation}
		\mathcal{E}_{\alpha}^{R}\left [\Psi_{N}\right ]\geq \frac{1}{b_{2}}\sum_{Q\in \mathcal{G}}\mathcal{E}^R _{Q}\left [\Psi_{N}\right ].\nn
	\end{equation}
	On each square $Q\in\mathcal{G}$ we have~\eqref{eq:enro} and may thus apply Theorem~\ref{th:LT_FN} to obtain 
	\begin{equation}
		\mathcal{E}_{\alpha}^{R}\left [\Psi_{N}\right ]\geq \frac{C^{\mathrm{FN}}\left |\alpha -1\right| }{b_{2}}\sum_{Q\in \mathcal{G}}\int_{Q}\rho^{2}_{\Psi_{N}}(\bx)\d \bx\geq \frac{C^{\mathrm{FN}}\left |\alpha -1\right|}{b_{2}}\int_{\R^{2}}\rho^{2}_{\Psi_{N}}(\bx)\d \bx\nn
	\end{equation}
	where we used the first inequality of~\eqref{ine:bes} in the last step.
	This provides the desired estimate with 
	\begin{equation}
		C^{\mathrm{EA}}=\min\left \{\frac{C^{\mathrm{FN}}}{b_{2}},\frac{1}{C_{2}N_<}\right \}.\nn
	\end{equation}
\end{proof}

\subsection{Reduction to a local exclusion principle}\label{sec:local LT}

We can now state the Local Pauli exclusion theorem we use to establish Theorem \ref{th:LT_FN}. Its proof will be the content of Section~\ref{sec:local}. 

\begin{theorem}[\textbf{Local exclusion principle for extended anyons}]\mbox{}\label{LEP_an}\\
	There exist three numbers $N_<$, $N_>$ and $C^{\mathrm{LE}}$ independent of $\alpha$ and $R$ such that if $\Psi_{N}\in L^{2}_{\mathrm{asym}}\left (\R^{2N}\right )$ is a $L^2$-normalized wave-function and $Q$ a square for which
	\begin{equation}\label{as:rho2}
		N_< \leq \int_{Q}\rho_{\Psi_{N}}\left (\bx\right )\d \bx\leq N_>
	\end{equation}
	we have that
	\begin{equation}\label{ine:le}
		\mathcal{E}^{R}_{Q}\left [ \Psi_{N}\right ]\geq C^{\mathrm{LE}}\frac{\left |\alpha -1 \right | }{\left \b Q\right \b} \int_{Q}	\rho_{\Psi_{N}}(\bx)\d \bx .
	\end{equation}
\end{theorem}

To prove Theorem \ref{th:LT_FN} we combine the above with a local uncertainty principle, i.e. essentially a Poincar\'e-Sobolev inequality. We use the version from ~\cite[Lemma~3.4]{Nam-20}, which is convenient for our purpose.

\begin{lemma}[\textbf{Local uncertainty}]\mbox{}\label{local_uncertainty}\\
	Let $\Psi_{N}\in H^{1}\left (\R^{2N}\right )$ for arbitrary $N\geq 1$ and let $Q$ be a square in $\R^{2}$. Then
	\begin{equation}\label{ine:li}
		\mathcal{E}^{R}_{Q}\left [\Psi_{N}\right ]\geq \frac{1}{4}\sum_{j=1}^{N}\int_{\R^{2N}}\left \b \nabla_{\bx_{j}} \b\Psi_{N}\b\right \b^{2}\1_{Q}\left (\bx_{j}\right ) \geqslant \frac{C_{2}\int_{Q}\rho^{2}_{\Psi_{N}}\left (\bx \right )\d \bx }{\int_{Q}\rho_{\Psi_{N}}\left (\bx \right ) \d \bx}-\frac{1}{\left \b Q\right \b}\int_{Q}\rho_{\Psi_{N}}\left (\bx \right )\d \bx
	\end{equation}
	for a universal constant $C_2$.
\end{lemma}
\begin{proof}
	We apply~\cite[Lemma 3.4]{Nam-20} to the wave function $\left \b\Psi_{N}\right \b$.
\end{proof}

\begin{proof}[Proof of Theorem \ref{th:LT_FN}, second implication in~\eqref{eq:structure}]
We assume~\eqref{as:rho2}. Combining Inequalities~\eqref{ine:le} and~\eqref{ine:li} we obtain, for any $\epsilon \in\left [0,1\right ]$ 
	\begin{align}
		\left (1-\epsilon+\epsilon\right )	\mathcal{E}^{R}_{Q}\left [\Psi_{N}\right ]&\geq \epsilon\frac{C_{2}\int_{Q}\rho^{2}_{\Psi_{N}}\left (\bx \right )\d \bx }{\int_{Q}\rho_{\Psi_{N}}\left (\bx \right ) \d \bx}-\frac{\epsilon}{\left \b Q\right \b}\int_{Q}\rho_{\Psi_{N}}\left (\bx \right )\d \bx\nn \\
		&+\left (1-\epsilon\right ) \frac{C^{\mathrm{LE}}|\alpha-1|}{\left \b Q\right \b}\int_{Q}	\rho_{\Psi_{N}}(\bx)\d \bx\nn \\
		&= \epsilon\frac{C_{2}\int_{Q}\rho^{2}_{\Psi_{N}}\left (\bx \right )\d \bx}{\int_{Q}\rho_{\Psi_{N}}(\bx)\d \bx}+\left ((1-\epsilon)C^{\mathrm{LE}}\b \alpha -1\b -\epsilon\right )\frac{1}{\b Q\b}\int_{Q}\rho_{\Psi_{N}}(\bx)\d \bx.\nn
	\end{align}
	We choose (clearly this is smaller than $1$)
	$$ \epsilon = \frac{C^{\mathrm{LE}}|\alpha-1|}{1+C^{\mathrm{LE}}|\alpha-1|}$$
	which makes the expression in parenthesis vanish. The desired result follows by bounding the remaining integral of $\rho_{\Psi_{N}}$ to obtain the constant
	\begin{equation}
		C^{\mathrm{FN}}=\frac{C_{2} }{N_>}\frac{C^{\mathrm{LE}}|\alpha-1|}{1+C^{\mathrm{LE}}|\alpha-1|}.\nn
	\end{equation}
\end{proof}

\section{Local exclusion principle for extended anyons}\label{sec:local}

There remains to deal with the heart of the matter, namely the proof of Theorem~\ref{LEP_an}. This result is true without any assumption on $\gamma =RL^{-1}$, but we use different proofs for three particular ranges of $\gamma$, thus covering all $\gamma \in \R_{+}$. We introduce two constants $c_{1}$ and $c_{2}>c_{1}$ to be fixed later on and work in three types of boxes, corresponding to the sketch of proof given at the end of Section~\ref{sec:intro}:
\begin{enumerate}
	\item \textbf{Large boxes} where $\gamma< c_{1}$
	\item \textbf{Medium boxes} where $c_{1}\leq \gamma \leq c_{2}$
	\item \textbf{Small boxes} where $\gamma> c_{2}$
\end{enumerate}
In each case we establish a lower bound on the energy of $n$ particles localized in $Q$, uniformly with respect to a number $m\geq 0$ of particles outside the box. We combine the three results in the end of the section to obtain the energy $\mathcal{E}_{Q}$ of \eqref{def:EQ}. The logic is that the first bound will be proven for $c_1$ small enough, the third bound for $c_2$ large enough, but the second bound is valid for all values of $c_1,c_2$, provided that $c_{1}$ be bounded away from zero and that $c_{2}$ remains bounded above.

The following notation is used throughout this section. We define the Neumann ground-state energy for $n$ extended anyons on a domain $Q \subset \R^{2}$ interacting with $m$ anyons in the exterior $Q^c=\R^{2}\setminus Q$.  We denote $Q_{0}=\left [0,1\right ]^{2}$ and define
\begin{equation}\label{eq:potbox}
	\bA^{R}_{j}\left (X_{n},Y_{m}\right )=\sum_{\substack{k=1\\k\neq j}}^{n}\nabla^{\perp}w_{R}\left (\bx_{j}-\bx_{k}\right ) + \sum_{k=1}^{m}\nabla^{\perp}w_{R}\left (\bx_{j}-\by_{k}\right )
\end{equation}
the magnetic vector potential for $n$ particles living in $Q$ (with coordinates $X_n = \left(\bx_{1},...,\bx_{n} \right)$) interacting with $m$ fixed anyons located at $Y_{m}=\left(\by_{1},...,\by_{m}\right)$, all outside $Q$. In line with the definition of $e_2$ in~\eqref{KEbis} we also denote
\begin{equation}\label{eq:scalbox}
 V^R _j (X_n,Y_m) = 2\pi\alpha\sum_{k\neq j}\frac{\1_{B\left (\bx_{k},R\right )}\left (\bx_{j} \right )}{\pi R^{2}} + 2\pi \alpha\sum_{k= 1} ^m \frac{\1_{B\left (\by_{k},R\right )}\left (\bx_{j} \right )}{\pi R^{2}}.
\end{equation}
For $\Psi_n \in L^2 (Q^n)$, set
\begin{align}\label{eq:enerinout}
\mathcal{E}^{R}_{n}\left (Q,Y_m\right )\left [\Psi_n\right ] &:= \frac{1}{2}\sum_{j=1}^{n}\int_{Q ^{n}}\left \b\left( -\im\nabla_{\bx_j}+\alpha	\bA^{R}_{j}\left (X_{n},Y_{m}\right )\right )\Psi_{n}\right \b^{2} \d X_{n} \nonumber\\
&+ \frac{1}{4} \sum_{j=1}^{n}\int_{Q ^{n}}\left \b\nabla_{\bx_j}|\Psi_{n}|\right \b^{2} \d X_{n} + \frac{1}{4}\sum_{j=1}^{n}\int_{Q ^{n}}V^R _j (X_n,Y_m)|\Psi_{n}|^{2} \d X_{n} \nonumber\\
&=:\sum_{j=1}^{n} \int_{Q ^{n}} e_{j}(\Psi_n,Y_m) \d X_{n}.
\end{align}
The fermionic Neumann energy is then
\begin{equation}
	E^{R}_{n}\left ( Q,m \right ):=\inf_{Y_{m}\in \left (Q^c \right )^{m}} \inf \left \{\mathcal{E}^{R}_{n}\left (Q,Y_m\right )\left [\Psi_n\right ] ,\Psi_{n}\in{L_{\rm asym}^{2}(Q^n) ,\; \norm{\Psi_n}_{L^{2}(Q^{n})} }=1 \right\}.
\end{equation}
We will drop the arguments $Q$ or $m$ of the previous energy when $Q=Q_{0}= [0,1]^2$ or $m=0$, namely 
$$ E^R_n (m) := E^{R}_{n}\left ( Q_0,m \right ) \mbox{ and } E^R_n := E^R_n \left ( Q_0,0 \right ).$$
We also introduce the notation 
	\begin{equation}
		D_{\bx_j}^{R}\left (Y_m\right ):=-\im\nabla_{\bx_j}+\alpha	\bA^{R}_{j}\left (X_{n},Y_{m}\right ).\nn
	\end{equation}
The main goal of this section is to prove 

\begin{proposition}[\textbf{Exclusion principle on finite boxes}]\mbox{}\label{lem:exclbox}\\
	There exist constants $\underline{N} >0$ and $C>0$ such that, for any $R > 0$, $m\geq 0$ and $N \geq \underline{N}$
	\begin{equation}
		E^{R}_{N}\left (Q,m\right )\geqslant \frac{C N \b 1- \alpha \b}{\b Q\b}.	
		\label{eq:exclmain}
	\end{equation}	
\end{proposition}

Before distinguishing between different box sizes in order to follow the strategy explained above, we start by general considerations that will reduce the proof to bounded particle numbers. We borrow and adapt several arguments from~\cite{LunSei-18}. 

\begin{lemma}[\textbf{Scaling property of the energy on a square}]\mbox{}\label{scaling_prop}\\
	For any square $Q$ such that $\left \b Q\right \b =L^{2}$ we have the scaling property
	\begin{equation}
		E^{R}_{n}\left (Q,m\right )=\left \b Q\right \b ^{-1}E^{R/L}_{n}\left (m\right ).
		\label{scaling}
	\end{equation}
\end{lemma}

\begin{proof}
	We first translate the variables to work on $[0,L]^2$, with $L$ the side-length of the original square. We proceed to the change of variables $\bx \to L\bx $. The first term in~\eqref{eq:enerinout} becomes
	\begin{equation}
		\sum_{j=1}^{n}\int_{Q_{0}^{n}}\left \b\left (\frac{-\im\nabla_{\bx_j}}{L}+\alpha\bA_{j}^{R}\left (LX_n,Y_m\right )\right )\Psi\left (L\bx_{1},...,L\bx_{n}\right )\right \b^{2}\left \b Q\right \b^{n}\d\bx_{1}...\d\bx_{n}.
	\end{equation} 
	Using the definition~\eqref{AJRr}, the property $\left \b L\bx\right \b_{R}=L\left \b \bx\right \b_{R/L}$ this becomes
	\begin{equation}\label{RL}
		\frac{1}{\left \b Q\right \b}\sum_{j=1}^{n}\int_{Q_{0}^{n}}\left \b\left (-\im\nabla_{\bx_j}+\alpha\bA_{j}^{R/L}(X_n,\frac{Y_{m}}{L})\right )\Phi\left (\bx_{1},...,\bx_{n}\right )\right \b^{2}\d\bx_{1}...\d\bx_{n}
	\end{equation}
	with $\Phi\left (\bx_{1},...,\bx_{n}\right ) =L^{n}\Psi\left (L\bx_{1},...,L\bx_{n}\right )$ so that $\int_{Q_{0}^{n}}\left \b\Phi\right \b^{2}=1$. The second part of~\eqref{eq:enerinout} is dealt with similarly and we conclude by taking the infimum.
\end{proof}

Next we have the equivalent of \cite[Lemma 4.2]{LunSei-18}. 

\begin{lemma}[\textbf{Superadditivity}]\mbox{}\label{superadd}\\
	Let $Q:=\cup_{q=1}^{K}Q_{q}$ with $\left \{Q_{q}\right \}_{q=1}^{K}$ a collection of disjoint and simply connected subsets of $Q$. Let $\vec{n}\in \mathbb{N}^{K}$ such that $\sum_{q=1}^{K}n_{q}=n$. We define the potential
	\begin{equation}
		W (X_n):=\sum_{\vec{n}}\sum_{q=1}^{K}E^{R}_{n_{q}}\left (Q_{q},m+n-n_{q}\right )\mathds{1}_{\vec{n}}\left (\bx_{1},\dots,\bx_{n}\right )\nn
	\end{equation}
	where $\1_{\vec{n}}$ denote the characteristic function of the subset of $Q^{n}$ where exactly $n_{q}$ of the points $\left \{\bx_{1},...,\bx_{n} \right\}$ are in $Q_{q}$ for all $1\leqslant q\leqslant K.$
	
	We have for any $Y_{m}\in \left (Q^{c}\right )^{m}$ and any normalized $\Psi_{n}\in L^{2}_{\mathrm{asym}}\left (Q^{n}\right )$ that
	\begin{equation}\label{modif_diam}
	\mathcal{E}^{R}_{n}\left (Q,Y_m\right )\left [\Psi_n\right ] \geqslant 	\int_{Q ^{n}}W\left (X_{n}\right )\left \vert\Psi_{n}\right \vert^{2}\d \bx_{1}\dots\d \bx_{n}.
	\end{equation}
	In particular
	\begin{equation}
		E^{R}_{n}\left (Q,m\right )\geqslant \min_{\vec{n}}\sum_{q=1}^{K}	E^{R}_{n_{q}}\left (Q_{q},m+n-n_{q}\right ).
		\label{division}
	\end{equation}
\end{lemma}

\begin{proof}	
	For any $Y_{m}\in \left (Q^{c}\right )^{m}$ and a normalized $\Psi_{n}\in L^{2}_{\mathrm{asym}}\left (Q^{n}\right )$ we have, with $e_{j}$ defined in \eqref{eq:enerinout}, that

	\begin{align*}
			\mathcal{E}^{R}_{n}\left (Q,Y_m\right )\left [\Psi_n\right ]&=
		\sum_{j=1}^{n}\int_{Q ^{n}}e_{j}(\Psi_{n},Y_m)\d \bx_{1}\dots\d \bx_{n}\\
		&=\sum_{\left \{A_{k}\right \}}\int_{Q_{1}^{\b A_{1}\b}}\int_{Q_{2}^{\b A_{2}\b}}...\int_{Q_{K}^{\b A_{K}\b}}\sum_{q=1}^{K}\sum_{j\in A_{q}}e_{j}(\Psi_{n},Y_m) \d \bx_{1}\dots\d \bx_{n}\\
		&\geqslant\sum_{\left \{A_{k}\right \}}\sum_{q=1}^{K}E^{R}_{\b A_{q}\b }\left( Q_{q},m+n-\b A_{q}\b\right)\int_{Q_{1}^{\b A_{1}\b}}\int_{Q_{2}^{\b A_{2}\b}}...\int_{Q_{K}^{\b A_{K}\b}}\left| \Psi_{n}\right|^{2}\d \bx_{1}\dots\d \bx_{n}\nonumber\\
			&=\sum_{\vec{n}}\sum_{q=1}^{K}E^{R}_{n_{q} }\left( Q_{q},m+n-n_{q}\right)\int_{Q^{n}}\mathds{1}_{\vec{n}}\left \b \Psi_{n}\right \b^{2}\d \bx_{1}\dots\d \bx_{n}\\
			&=\int_{Q ^{n}}W\left (X_{n}\right )\left| \Psi_{n}\right|^{2}\d \bx_{1}\dots\d \bx_{n}.
	\end{align*}
	Here the sum over $\left \{A_{k}\right \}$ runs over all partitions of the particles into the sets $Q_{q}$, $\mathrm{i.e}$, over collections of disjoint subsets $A_{k}\subset \left \{1,2,\dots ,N\right \}$ such that $\b A_{1}\b +\dots +\b A_{K}\b =n$. We used the fact that (by the definition of $E_{n}^{R}\left (Q,m\right )$) we have for almost every $\left \{\bx_{l}\right \}_{l\in A_{q}^{c}}\in\left (Q_{q}^{c}\right )^{n-n_{q}}$ that
	\begin{equation}
	\sum_{j\in A_{q}}\int_{Q_{q}^{\vert A_{q}\vert}}e_{j}\left (\Psi_{n},Y_{m}\right )\d X_{A_{q}}\geq E^{R}_{\vert A_{q}\vert}\left (Q_{q},m+n-n_{q}\right )\int_{Q_{q}^{\vert A_{q}\vert}}\left \vert \Psi_{n}\right \vert^{2}\d X_{A_{q}}
	\end{equation}
where we denoted
$$ \d X_{A_{l}}=\prod_{j\in A_{l}}\d \bx_{j}.$$
This proves the first bound of the lemma and the second follows by bounding $W$ from below by its smallest value.
\end{proof}

Now we reduce to bound with unifnormly bounded particle numbers.
%
%
Denote 
\begin{equation}\label{eq:lowest}
E (n,Q) := \inf_{R>0,m  \in \mathbb{N}} E^{R}_{n}\left(Q,m\right).  
\end{equation}

\begin{lemma}[\textbf{Reduction to finite particle numbers}]\label{lem:finite}\mbox{}\\
Let $\overline{N} = 4^k$ for some integer $k\geq 1$ and $N \geq \overline{N} $. Assume that $E(n,Q) >0$ for all $4^{k-1}+1 \leq n\leq 4N$. There exists a constant $C_k>0$, independent of $N$, such that
\begin{equation}\label{eq:reduc finite}
E (N,Q) \geq C_k N \min\left(E(4^{k-1}+1,Q), \ldots, E(4^k,Q)\right). 
\end{equation}
\end{lemma}

The assumption $E(n,Q) >0$ for all $4^{k-1}+1 \leq n\leq 4N$ is only temporary. Its validity will follow from the considerations in the next subsections.

\begin{proof}
 This follows exactly the proof of~\cite[Lemma~4.8]{LunSei-18}. In view of Lemma~\ref{scaling_prop} we work on the unit square and drop $Q$ from the notation. Splitting it into four equally large squares and using the superadditivity of Lemma~\ref{superadd} and the scaling property of Lemma~\ref{scaling_prop} we have 
 $$ E(N) \geq 4 \min_{\vec{n}} \sum_{q=1} ^4 E(n_q).$$
 It is here important that we have taken the infimum over $R$ and $m$ in~\eqref{eq:lowest}. At least one of the squares must contain $N/4$ particles, hence, dropping the other terms, we obtain 
 \begin{equation}\label{eq:iterate}
  E(N) \geq 4 \min_{j = \lfloor \frac{N}{4} \rfloor + 1 , \ldots ,N} E(j).
 \end{equation}
  Denoting
 $$ e_i:= \min\left\{ E (j), 4^{i} +1 \leq j \leq 4^{i+1} \right\}$$
 we obtain from~\eqref{eq:iterate} that 
 $$ e_i \geq 4 \min (e_i,e_{i-1})$$
 and, since our assumption implies $e_i > 0$, 
 \begin{equation}\label{ine:en}
 	e_i \geq 4 e_{i-1}.
 \end{equation}
 We treat the case $N= 4^l$ for $l$ integer for simplicity (the generalization is straightforward, as in~\cite[Lemma~4.8]{LunSei-18}). We then have
 \begin{equation}
 E\left (N\right )\geq 4e_{l-1}\nn
 \end{equation}
 and iterating \eqref{ine:en} a finite number of times, we deduce
 $$ E(N) \geq 4^{l-k} \min\left(E(4^{k-1}+1), \ldots, E(4^k)\right) \geq C_k N \min\left(E(4^{k-1}+1), \ldots, E(4^k)\right)$$
 where $C_k$ depends only on $k$, thus on $\overline{N}=4^{k}$.
\end{proof}

\subsection{Large boxes}\label{sec:large}

Our exclusion principle on large boxes reads as follows: 

\begin{proposition}[\textbf{Exclusion principle on large boxes}]\mbox{}\label{lem:excllarge}\\
	There exist $0 < c_1 < \frac{1}{12}$, a natural number $\overline{N}>0$ and $C>0$ independent of $R$, $\alpha$, $Q$, $m$ and $n$ such that, for any $\gamma = R/L \leq c_1 $ and $2 \leq n\leq \overline{N}$ 
	\begin{equation}
		E^{R}_{n}\left (Q,m\right )\geqslant \frac{C \b 1- \alpha \b}{\b Q\b}.	
		\label{eq:gage1}
	\end{equation}
	
\end{proposition}

The proof is inspired from~\cite{LunSei-18}. It occupies the rest of the subsection and requires several preparatory lemmas. 

\begin{lemma}[\textbf{Upper-bound for the two-particles energy on the unit square}]\mbox{}\label{lem:upp}\\   The two-particles energy on the unit square can be bounded by a constant $C$ which does not depend on $R$, $m$ and $\alpha$,
	\begin{equation}
		E^{R}_{2}\left (m\right )\leqslant C\nn
	\end{equation}

\end{lemma}
\begin{proof}
	We can construct a fermionic trial state on the unit cube $Q_{0}$ by setting  
	\begin{equation}
		\psi_{2}\left (\bx_{1},\bx_{2}\right ):=\left (x_{1}^{1}-x_{2}^{1}\right )+\left (x_{1}^{2}-x_{2}^{2}\right )\nn
	\end{equation}
	with the notation $\bx_{j}=\left (x_{j}^{1},x_{j}^{2}\right )$.
	We then have 
	\begin{align}
		\norm{\psi_{2}}_{L^{2}\left (Q_{0}^{2}\right )}E^{R}_{2}\left (m\right )&\leq \mathcal{E}^{R}_{2}\left [	\psi_{2}\right ]=2\int_{Q_{0}^{2}}\left \b -\im\nabla_{\bx_{1}}\psi_{2}+\alpha\frac{\left ( \bx_{1}-\bx_{2}\right )^{\perp}}{\b  \bx_{1}-\bx_{2}\b_{R}^{2}}\psi_{2} \right  \b^{2}\d \bx_{1}\d \bx_{2}\nn\\
		&= 2\int_{Q_{0}^{2}}\left (\left \b \nabla_{\bx_{1}}\psi_{2}\right  \b^{2} + \alpha^{2}\frac{\b \psi_{2}\b^{2} }{\b  \bx_{1}-\bx_{2}\b^{2}}\right )\d \bx_{1}\d \bx_{2}\nn\\
		&\leq 8\nn
	\end{align}
	where we expanded the square using that $\psi_{2}$ is real and applied both that $\b\alpha\b \leq 1$ and that $\b  \bx_{1}-\bx_{2}\b_{R}\geq \b\bx_{1}-\bx_{2}\b$.
\end{proof}

A key step is to prove a lower bound on the energy of $n$ particles in terms of the energy with exactly two particles. The basic idea is simple: if the ground state for $n$ particles is nearly constant, it assigns non-zero probability to the event ``a subsquare contains exactly two particles''. If the ground state is not nearly constant, it must come with some positive kinetic energy.

\begin{lemma}[\textbf{A priori bounds in terms of $	E^{R}_{2}$}]\mbox{}\label{lem:ap_E2}\\  
	There exist two constants $C_{1}$ and $C_{2}$ independent of $R$, $\alpha$, $n$ and $m$ such that
	\begin{equation}\label{eq:red}
		E^{R}_{n}\left (m\right )\geqslant \frac{C_{n}}{C_{1}+C_{2}C_{n}}	E^{2R}_{2}\left (m+n-2\right )
	\end{equation}
	where $C_{n}=\binom{n}{2}\left (\frac{3}{4}\right )^{n-2}$.
\end{lemma}
\begin{proof}
	We follow the same route as in~\cite[Lemma 4.3]{LunSei-18} but with fermionic wave functions becoming bosonic wave functions via the diamagnetic inequality.
	
	We start from $	E^{R}_{n}\left (m\right )$ and divide $Q_{0}$ of side-length $1$ in four equally large squares of side-length $1/2$ 
	$$Q_{0}:=Q_{1}\sqcup Q_{2}\sqcup Q_{3}\sqcup Q_{4}.$$
	We apply Lemma \ref{superadd} to obtain
	\begin{align}
		\sum_{j=1}^{N}\int_{Q_{0}^{n}}e_{j}(\Psi,Y_m)\geqslant\int_{Q_{0}^{n}}W\left \vert \Psi\right \vert^{2}\nn
	\end{align}
	and we use the scaling property \eqref{scaling} to get
	\begin{equation}
		W\geqslant	W_{2}=4E^{2R}_{2}\left (m+n-2\right )\sum_{\vec{n}}\sum_{q=1}^{4} \delta_{n_{q}=2}\mathds{1}_{\vec{n}}\nn .
	\end{equation}
 We can compute 
	\begin{align}
		\int_{Q_{0}^{n}}W_{2}&=E_{2}^{R}\left (m+n-2\right )\binom{n}{2}\left (\frac{3}{4}\right )^{n-2}=C_{n}E_{2}^{R}\left (m+n-2\right )\nn
	\end{align}
	by counting the probability that exactly two particles are in a given square. We now want to estimate
	\begin{equation}
		\int_{Q_{0}^{n}}W_{2}\left \vert \Psi\right \vert.\nn
	\end{equation}
	To this aim we will use a little bit of the kinetic energy 
	\begin{equation}
		T_{\alpha}^{R,m}=\sum_{j=1}^{n}\int_{Q_{0} ^{n}}\left \b\left( -\im\nabla_{\bx_j}+\alpha	\bA^{R}_{j}\left (X_{n},Y_{m}\right )\right )\Psi_{n}\right \b^{2} \d X_{n}\nn
	\end{equation} and the inequality \eqref{modif_diam} to obtain
	\begin{equation}
		T_{\alpha}^{R,m}=\kappa	T_{\alpha}^{R,m}+\left (1-\kappa\right )T_{\alpha}^{R,m}\geqslant \kappa	T_{\alpha}^{R,m}+\left (1-\kappa\right )W_{2}\nn
	\end{equation}
	for any $\kappa\in \left [0,1\right ]$. The diamagnetic inequality \cite[Theorem 7.21]{LieLos-01} leads to
	\begin{align}
		E_{n}^{R}\left (m\right )&=\inf_{\Psi\in L^{2}_{\rm asym}\left (Q_0^{n}\right ), \norm{\psi}=1 }\left \langle \Psi ,T_{\alpha}^{R,m}\Psi \right \rangle\nonumber\\
		&\geqslant \inf_{\Psi\in L^{2}_{\rm asym}\left (Q_0^{n}\right ),  \norm{\psi}=1 }\left \langle \left \vert\Psi\right \vert ,\left (-\kappa\Delta_{Q_0^{n}}+\left (1-\kappa\right )W_{2}\right )\left \vert\Psi\right \vert \right \rangle\nonumber\\
		&\geqslant  \inf_{\Psi\in L^{2}_{\rm sym}\left (Q_0^{n}\right ),  \norm{\psi}=1 }\spec\left ( -\kappa\Delta_{Q_0^{n}}+\left (1-\kappa\right )W_{2}\right )\label{op_on_bos}
	\end{align}
	where $\Delta_{Q_{0}^n}$ is the Neumann Laplacian on $Q_{0}^{n}$. Note that $E_{n}^{R}\left (Q_{0},m\right )$ was defined on fermionic wave functions whereas now  the operator of \eqref{op_on_bos} acts on bosons. We consider the orthognal projection 
	\begin{equation}
		P:= | u_{0} \rangle \langle u_0 |\nn
	\end{equation}
	onto the normalised ground state of $\Delta_{Q_{0}^{n}}$, i.e. the constant function $u_{0}\equiv 1$, and the orthogonal complement 
	\begin{equation}
		P^{\perp}:=\mathds{1}-P\nn
	\end{equation}
	for which we have
	\begin{equation}
		-\Delta_{Q_0^{n}}\geqslant \pi^{2}P^{\perp}.\nn
	\end{equation}
	We use the Cauchy-Schwarz inequality to get that
	\begin{equation}
		W_{2}=\left (P+P^{\perp}\right )W_{2}\left (P+P^{\perp}\right )\geqslant \left (1-\epsilon\right )PW_{2}P+\left (1-\epsilon^{-1}\right )P^{\perp}W_{2}P^{\perp}\nn
	\end{equation}
	for arbitrary $\epsilon\in \left [0,1\right ]$.
	
	Now we have
	\begin{equation}
		PW_{2}P=P\int_{Q_{0}^{n}}W_{2}\nn
	\end{equation}
	and 
	\begin{equation}
		P^{\perp}W_{2}P^{\perp}\leqslant P^{\perp}\norm{W_{2}}_{\infty}\leqslant16E_{2}^{2R}\left (m+n-2\right )P^{\perp}.\nn
	\end{equation}
	The combination of the previous estimates leads to
	\begin{align}
		-\kappa\Delta_{Q_0^{n}}+\left (1-\kappa\right )W_{2}\geqslant \left (\kappa\pi^{2} -\left (1-\kappa\right )\left (\epsilon^{-1}-1\right )16E_{2}^{2R}\left (m+n-2\right )\right )P^{\perp}\nn\\
		+\left (1-\kappa\right )\left (1-\epsilon\right )C_{n}E_{2}^{2R}\left (m+n-2\right )P.\nn
	\end{align}
	We choose $\kappa$ to make the prefactors in front of the projections equal
	\begin{equation}
		\kappa =\frac{E_{2}^{2R}\left (m+n-2\right )\left [C_{n}\left (1-\epsilon\right )+16\left (\epsilon^{-1}-1\right )\right ]}{\pi^{2}+E_{2}^{2R}\left (m+n-2\right )\left [C_{n}\left (1-\epsilon\right )+16\left (\epsilon^{-1}-1\right )\right ]}.\nn
	\end{equation}
	We finally obtain the bound
	\begin{equation}
		E^{R}_{n}\left (m\right )\geqslant \frac{\pi^{2}C_{n}\left (1-\epsilon\right )E_{2}^{2R}\left (m+n-2\right )}{\pi^{2}+E_{2}^{2R}\left (m+n-2\right )\left [C_{n}\left (1-\epsilon\right )+16\left (\epsilon^{-1}-1\right )\right ]}\nn
	\end{equation}
	where the choice $\epsilon =\frac{1}{2}$ combined with the upper bound on $E_{2}^{R}\left (m+n-2\right )$ of Lemma \ref{lem:upp} gives the result.

\end{proof}

The above lemma provides a lower bound in terms of the energy with exactly two particles inside the box, but possibly with the extra influence of many fixed particles outside the box. For $R=0$, the influence of the latter can be gauged away as in ~\cite{LunSei-18}. In our case, the magnetic flux-tube they carry might overlap the box. We will show that this interaction can be gauged away, at a controlable cost. Indeed, since $R \ll L$ the flux-tube of outside particles touches the particles in the box only when they are close to the boundary.

\begin{lemma}[\textbf{Gauging away particles outside the box}]\mbox{}\\\label{lem:toe2}
	There exists a small $\gamma_{0}< \frac{1}{8}$ such that for any $\gamma \leq\gamma_{0} $ we can find a constant $C_{1}$ independent of $R$, $m$, $\alpha$ and $Q$ and a constant $C_{\gamma_{0}}$ only depending on $\gamma_{0}$ such that
	\begin{equation}
		E^{R}_{2}\left (Q,m\right )\geqslant \min\left \{C_{1}E_{2}^{\tilde{R}}\left (Q\right ),\frac{C_{\gamma_{0}}}{\b Q\b}\right \} 	
		\label{eq:gage}
	\end{equation}
	where $\tilde{R}=LR/(L-4R)$.
\end{lemma}
\begin{proof}
Here we drop the second term in~\eqref{eq:enerinout} to consider only the magnetic kinetic energy. We remove from $Q$ a very thin corridor such that the particles (of small radius) inside the restricted domain cannot interact with the outside. The proof will show that if the density is small in the corridor, we can neglect it. If not then the kinetic energy has to be large enough for the statement to hold.
	
	We assume that $L\geq 8R$ and define $S=\left [2R, L-2R\right ]^{2}$. On $S^{2}$, define the change of gauge
	\begin{equation}
		\widetilde{\Psi}_{2}=\prod_{k=1}^{m}e^{\im\alpha \left (\phi_{1k}+\phi_{2k}\right )}\Psi_{2}\;\;
		\text{with}\;\;
		\phi_{jk}:=\arg \frac{\bx_{j}-\by_{k}}{\b \bx_{j}-\by_{k}\b}
	\end{equation}
	where $\by_1,\ldots,\by_m$ are the coordinates of particles outside the box. Note that 
	$$ \nabla_{\bx_j}\phi_{jk} = \frac{\left (\bx_{j}-\by_{k}\right )^{\perp}}{\left \b \bx_{j}-\by_{k}\right \b^{2}}= \frac{\left (\bx_{j}-\by_{k}\right )^{\perp}}{\left \b \bx_{j}-\by_{k}\right \b^{2}_{R}}$$
	when $|\bx_j - \by_k| \geq R$, and that $\nabla_{\bx_j}\phi_{jk}$ is regular and curl-free in that range. We thus have, using also the fermionic symmetry,
	\begin{align}
		\mathcal{E}^{R}_{2}\left (Q,Y_m\right )[\Psi_2]&\geq \int_{Q^{2}}\left \b\left (-\im\nabla_{1}+\alpha\frac{\left (\bx_{1}-\bx_{2}\right )^{\perp}}{\left \b \bx_{1}-\bx_{2}\right \b^{2}_{R}}+\alpha \sum_{k=1}^{m}\frac{\left (\bx_{1}-\by_{k}\right )^{\perp}}{\left \b \bx_{1}-\by_{k}\right \b^{2}_{R}}\right )\Psi_{2}\right \b^{2}\d \bx_{1}\d \bx_{2}\nonumber\\
		&\geq\int_{Q^{2}}\left \vert\left (-\im\nabla_{1}+\frac{\left (\bx_{1}-\bx_{2}\right )^{\perp}}{\left \vert \bx_{1}-\bx_{2}\right \vert^{2}_{R}}\right )\widetilde{\Psi}_{2}\right \vert^{2}\1_{S^{2}}\d \bx_{1}\d \bx_{2}\nonumber\\
		&\geq \frac{1}{2}E^{R}_{2}\left (S\right )\int_{S^{2}}\vert\Psi_{2}\vert^{2} \d \bx_{1}\d \bx_{2}
		\label{eq:last}
	\end{align}
	where we used that $\b\widetilde{\Psi_2}\b =\b\Psi_2\b $. We denote $\Omega = Q\backslash S$ with
	\begin{equation}\label{eq:corridor}
	\b \Omega \b =8 LR - 16 R^2 \leq 8 LR.
	\end{equation}
We have that 
\begin{equation}
\int_{Q^{2}}\vert\Psi_{2}\vert^{2}=\int_{S^{2}}\vert\Psi_{2}\vert^{2}+\int_{S\times \Omega}\vert\Psi_{2}\vert^{2}+\int_{Q\times \Omega}\vert\Psi_{2}\vert^{2}
\end{equation}
which implies that
\begin{equation}
	\int_{Q\times \Omega}\vert\Psi_{2}\vert^{2}\leq \int_{Q^{2}}\vert\Psi_{2}\vert^{2}-\int_{S^{2}}\vert\Psi_{2}\vert^{2}\leq 2\int_{Q\times \Omega}\vert\Psi_{2}\vert^{2}
\end{equation}
	We then observe that if for a small number $c< 1$
	\begin{equation}
		\int_{Q\times\Omega}\vert\Psi_{2}\vert^{2}\d \bx_{1}\d \bx_{2}\leq c\int_{Q^{2}}\vert\Psi_{2}\vert^{2}\d \bx_{1}\d \bx_{2}\nn
	\end{equation}
then
\begin{equation}
	\int_{S^{2}}\vert\Psi_{2}\vert^{2}\geq \left (1-2c\right )\int_{Q^{2}}\vert\Psi_{2}\vert^{2}
\end{equation}	and the first inequality of \eqref{eq:gage} results from \eqref{eq:last} to which we apply the scaling property~\eqref{scaling} leading to $C_{1}=\frac{1}{2}-c$.
	
	On the other hand, when 
	\begin{equation}\label{eq:cas2}
		\int_{Q\times\Omega}\vert \Psi_{2}\vert^{2}\d \bx_{1}\d \bx_{2}\geq c\int_{Q^{2}}\vert \Psi_{2}\vert^{2}\d \bx_{1}\d \bx_{2}
	\end{equation}
	we introduce $P$, the orthogonal projector on the constant function and $P^{\perp}$ such that 
	$$P+P^{\perp}=\1_{L^{2}\left (Q^{4}\right )}.$$
	Then 
	\begin{align}
		\int_{Q\times \Omega}\vert \Psi_{2}\vert^{2}\d \bx_{1}\d \bx_{2}
		&\leq 2\frac{\b Q\times \Omega\b}{\b Q\times Q\b}\int_{Q^{2}}\vert P\b \Psi_{2}\b\vert^{2}\d \bx_{1}\d \bx_{2} +2\int_{Q\times \Omega}\vert P^{\perp}\b \Psi_{2}\b\vert^{2}\d \bx_{1}\d \bx_{2}\label{eq:lastin}\nonumber\\
		&\leq 16\frac{L^{3}R}{L^{4}}\int_{Q^{2}}\vert P\b \Psi_{2}\b\vert^{2}\d \bx_{1}\d \bx_{2} +2\b Q\times \Omega\b^{\frac{1}{2}}\left (\int_{Q^{2}}\vert P^{\perp}\b \Psi_{2}\b\vert^{4}\d \bx_{1}\d \bx_{2}\right )^{\frac{1}{2}}\nn
	\end{align}
	where we used the Cauchy-Schwarz inequality and~\eqref{eq:corridor}. We now apply  a Poincar\'e-Sobolev inequality~\cite[Theorem 8.11]{LieLos-01} to $P^{\perp}\b \Psi_{2}\b=\b \Psi_{2}\b-P\b \Psi_{2}\b$
	\begin{equation}
		\left (\int_{Q^{2}}\vert P^{\perp}\b \Psi_{2}\b\vert^{4}\d \bx_{1}\d \bx_{2}\right )^{\frac{1}{2}}\leq C_{\mathrm{ps}} \int_{Q^{2}}\vert \nabla\b \Psi_{2}\b\vert^{2}\d \bx_{1}\d \bx_{2}\nn
	\end{equation}
	which provides 
	\begin{align}
		C_{\mathrm{ps}}\int_{Q^{2}}\vert \nabla\b \Psi_{2}\b\vert^{2}\d \bx_{1}\d \bx_{2}&\geq \frac{C}{L\sqrt{LR}}	\int_{Q\times \Omega}\vert \Psi_{2}\vert^{2}\d \bx_{1}\d \bx_{2}-\frac{CR}{ L^{2}\sqrt{LR}}\int_{Q^{2}}\vert \Psi_{2}\vert^{2}\d \bx_{1}\d \bx_{2}\label{eq:av}\\
		&\geq \frac{C}{L^{2}}\int_{Q^{2}}\vert \Psi_{2}\vert^{2}\d \bx_{1}\d \bx_{2}\left (c \left ( 6 \sqrt{\gamma}\right )^{-1}-\sqrt{\gamma}\right )\nn
	\end{align}
	using that $\vert \Omega\vert\geq 6RL$ and Assumption~\eqref{eq:cas2} for the first term in the right hand side of \eqref{eq:av}.
	The diamagnetic inequality \cite[Theorem 7.21]{LieLos-01} allows to obtain the second bound of \eqref{eq:gage} provided that we chose a $\gamma $ small enough.
\end{proof}

We are now reduced to finding a lower bound on the energy of two particles inside the box, with no particles outside. This can be estimated using previous results of~\cite{LarLun-16}. Essentially, with two isolated particles the effect of the flux tube between them may never be spoiled by that of other particles, and it provides a non-zero energy via a Hardy-type inequality. 

We need to introduce the function 
\begin{align}
	g^{2}:\R_{+}\times \left [0,1\right ]&\to \R\nn\\
	\nu ,\gamma &\mapsto g^{2}\left (\nu ,\gamma\right )\nn
\end{align}
which is the smallest positive solution $\lambda$ associated with the Bessel equation
\begin{equation}\label{eq:bessel}
	-u''-\frac{u'}{r}+\nu^{2}\frac{u}{r^{2}}=\lambda u
\end{equation}
on the interval $\left [\gamma,1\right ]$ with Neumann boundary conditions while $g\left (\nu ,\gamma\right )=\nu$ for $\gamma \geq 1$.
The function $g$ has the property that
\begin{equation}\label{lim:g}
	g\left (\nu ,\gamma\right )\xrightarrow[\gamma\to 0]{}j_{\nu}'\geq \sqrt{2\nu}
\end{equation}
where $j_{\nu}'$ denotes the first positive zero of the derivative of the Bessel function $J_{\nu}$.

\begin{lemma} [\textbf{Bound on $E_{2}^{R}$}]\mbox{}\\\label{Lem:E2}There exist $\gamma_{1}<1/12$ and a constant $C_{\gamma_{1}}$ only depending on $\gamma_{1}$ such that for all $\gamma \leq \gamma_{1}$ we have
	\begin{equation}\label{eq:E2}
		E_{2}^{\gamma} \geqslant C_{\gamma_{1}}\left \b \alpha -1\right \b .
	\end{equation}
	
\end{lemma}
\begin{proof}
	We apply~\cite[Lemma 5.3]{LarLun-16}. Note that this result is stated for anyons based on a bosonic wave function while in our case we used a fermionic one~\eqref{eq:stat}. However, the proof given in~\cite{LarLun-16} may be applied  to this case. Indeed, the basic building block is~\cite[Lemma~3.1]{LarLun-16}. We use it in the antipodal-antisymmetric case rather than the antipodal-symmetric case, and otherwise follow  the rest of the proof mutatis mutandis. The result is  
	\begin{equation}
		E_{2}^{R} \geqslant \frac{\pi}{48}g^{2}\left (c\alpha_{2} ,12\gamma\right )\left (1-12\gamma\right )^{3}_{+}\nn
	\end{equation}
	where 
	\begin{equation}\label{eq:alphaN}
		\alpha_{N}:=\min_{p\in\left \{0,1,\cdots,N-2\right\} }\min_{q\in \mathbb{Z}}\left \b\left (2p +1\right )\left (1-\alpha\right ) -2q\right \b .\nn
	\end{equation}
	We use that $\alpha_{2}=\b 1-\alpha\b$ for $\alpha \in [0,2]$. Combining with the limit property~\eqref{lim:g} we can pick a $\gamma_{1}$ such that
	\begin{equation}
		E_{2}^{R} \geqslant C_{\gamma_{1}}\left \b 1-\alpha \right \b .\nn
	\end{equation}
\end{proof}

We may now conclude the

\begin{proof}[Proof of Proposition~\ref{lem:excllarge}]
	We start by applying Lemma \ref{lem:ap_E2} with the scaling property \eqref{scaling} to the energy
	\begin{equation}
		E^{R}_{n}\left (Q,m\right )\geqslant \frac{C_n}{\b Q\b}	E^{2\gamma}_{2}\left (m+n-2\right ).\nn
	\end{equation}
	We use Lemma \ref{lem:toe2} to bound the two-particles energy appearing in the above. We then restrict $2\gamma \leq \gamma_{0}$ and obtain that
	\begin{equation}
		E^{R}_{n}\left (Q,m\right )\geqslant \frac{C_n}{\b Q\b}\times \left\{C_{1}E_{2}^{\frac{2\gamma}{1-4\gamma}}\;\;\text{or }\;\; C_{\gamma_{0}}\right\}.\nn
	\end{equation}
	If we consider the $C_{\gamma_{0}}$ case there is nothing more to prove and the final constant is $C_{\gamma_{0}}C_n$. In the other case we use Lemma \ref{Lem:E2} which provides a new restriction $2\gamma/(1-4\gamma)\leq \gamma_{1}$ under which the bound \eqref{eq:E2} leads to
	\begin{equation}
		E^{R}_{n}\left (Q,m\right )\geqslant \frac{C_n}{\b Q\b}C_{1}C_{\gamma_{1}}\b 1-\alpha\b.\nn
	\end{equation}
	This concludes the proof, upon redefining $C_n$ and taking 
	$$c_{1}=\frac{1}{2}\min\left \{\gamma_{0},\gamma_{1}/(2+4\gamma_{1})\right \}.$$
	The final constant is not uniform in $n$. Our assumption that $n\leq \overline{N}$ is needed to control it.
\end{proof}

\subsection{Medium Boxes}

We now deal with medium boxes where $c_{1}\leq \gamma \leq c_{2}$ with $c_{1}<\sqrt{2}\leq c_{2}$. In this case~\cite[Lemma 5.1]{LarLun-16} provides the desired bound for $\alpha$ separated from $0$ (i.e. the fermionic end, in our convention). We combine this with a ``perturbative'' treatment of the magnetic interaction for small $\alpha$ to obtain the 

\begin{proposition}[\textbf{Exclusion principle on medium boxes}]\label{lem:medium}\mbox{}\\
	Let $c_{1}\leq \gamma = R/L \leq c_{2}$ and $Q$ be a square of side-length $L$. Assume that $\underline{N}\leq n \leq \overline{N}$ for some large enough $\underline{N}$. There is a constant $C\left (c_{1},c_{2}\right )$ depending only on $c_1,c_2$ such that
	\begin{equation}
		E_{n}^{R}\left (m,Q\right )\geq C\left (c_{1},c_{2}\right )\frac{\b 1- \alpha \b}{\left \b Q\right \b}.\nn
	\end{equation} 
\end{proposition}

\begin{proof}
We simplify notations by working on the unit square and separate, as announced, the case of small $\alpha$ from the rest of the argument. 

\medskip

\noindent\textbf{Step 1, a bound linear in $|\alpha|$}. We introduce the function
\begin{equation}
	K_{\alpha}=\sqrt{2\vert\alpha\vert}\frac{I_{0}\left (\sqrt{2\vert\alpha\vert}\right )}{I_{1}\left (\sqrt{2\vert\alpha\vert}\right )}\nn
\end{equation}
where $I_{\nu}$ is the modified Bessel function of order $\nu$. One can show that 
\begin{equation}
	CI_{0}\left (2\right )\geqslant K_{\alpha}\geqslant 2\nn
\end{equation}
when $\alpha \in\left [0,2\right ]$. Indeed, the second bound comes from \cite[Lemma 5.1]{LarLun-16} while the first one follows from the fact that $I_{0}$ is an increasing function and that for $x \in\left [0,2\right ]$,
\begin{align}
	I_{1}\left (x\right )&:=\sum_{m=0}^{\infty}\frac{1}{m!\Gamma\left (m+2\right )}\left (\frac{x}{2}\right )^{2m+1}\geqslant \frac{x}{2\Gamma\left (2\right )}.\nn
\end{align}
We use~\cite[Lemma 5.1]{LarLun-16}, which provides two bounds depending on the range of $\gamma$, stated for bosonic based anyons. The proof starts from~\eqref{eq:basic mag}, so that lower bounds are obtained in terms of the modulus of the wave-function, which is always bosonic. One can thus follow the argument mutatis mutandis in our case.

	The first bound of the lemma holds when $\gamma <\sqrt{2}$ and is obtained by the application of Dyson's lemma on the kinetic energy added to the magnetic interaction energy obtained via~\eqref{eq:basic mag} (see~\cite[Lemma 1.1]{LarLun-16} for details). It states that
	\begin{equation}
		E_{n}^{R}\left (m\right )\geq \frac{\b \alpha \b \min\left \{\left (1-\gamma^{2}/2\right )^{-1},K_{\alpha}/2\right \}}{K_{\alpha}+2\b \alpha \b\left (-\ln\left (\gamma/\sqrt{2}\right )\right )}\left (n-1\right )_{+}.\nn
	\end{equation} 
	Under the additional assumption that $c_{1}\leq \gamma$ the divergence of the logarithm is under control and the bound reduces to
	\begin{equation}
		E_{n}^{R}\left (m\right )\geq C\frac{\b \alpha \b }{I_{0}\left (\sqrt{2}\right )+\left (-\log \left (c_{1}/\sqrt{2}\right )\right )}\left (n-1\right )_{+}\geq  C_{1}\b \alpha \b\left (n-1\right )_{+}.\nn
	\end{equation} 
	The second bound is valid for any $\gamma \geq \sqrt{2}$ and is obtained using the magnetic interaction of \cite[Lemma 1.1]{LarLun-16} where the indicator function equals $1$ on the whole box $Q$:
	\begin{equation}
		E_{n}^{R}\left (m\right )\geq  2\b \alpha \b \gamma^{-2}n\left (n-1\right )_{+}\geq \frac{2\b \alpha \b }{c_{2}^{2}}\left (n-1\right )_{+}.\nn
	\end{equation}
	These bounds yield the desired conclusion when $|\alpha| \geq c >0$ for a fixed constant $c>0$.
	There remains to deal with the case where $|\alpha|$ is allowed to become small. We will have to use that our basic wave-functions are fermionic, a constraint we dropped in the argument above (indeed, we have used only the second term in~\eqref{eq:enerinout}). 
	
	\medskip
	
	\noindent\textbf{Step 2, small $|\alpha|$.} We use the first term in~\eqref{eq:enerinout} and the fact that we work with antisymmetric functions. We split the vector potential~\eqref{eq:potbox} between the part generated by particles outside the box and particles inside the box : 
	\begin{align}\label{eq:inpot}
	 \bA^{\rm ex}_{j}\left (X_{n},Y_{m}\right) &= \bA^{\rm ex}_{j}\left (\bx_j,Y_{m}\right) := \sum_{k=1}^{m}\nabla^{\perp}w_{R}\left (\bx_{j}-\by_{k}\right )\nonumber\\
	 \bA^{\rm in}_{j}\left (X_{n},Y_{m}\right) &= \bA^{\rm in}_{j}\left (X_n\right) := \sum_{\substack{k=1\\k\neq j}}^{n}\nabla^{\perp}w_{R}\left (\bx_{j}-\bx_{k}\right).
	\end{align}
	The energy of a $L^2$-normalized antisymmetric wave-function $\Psi_n$ is then bounded from below as  
	\begin{align}\label{eq:smallal}
	\mathcal{E}^{R}_{n}\left (Q,Y_m\right )\left [\Psi_n\right ] &\geq \frac{1}{2}\sum_{j=1}^{n}\int_{Q ^{n}}\left \b\left( -\im\nabla_{\bx_j}+\alpha	\bA^{\rm ex}_{j}\left (\bx_j,Y_{m}\right) +\alpha	\bA^{\rm in}_{j}\left (X_n\right )\right )\Psi_{n}\right \b^{2} \d X_{n} \nn \\
	&\geq  \frac{1}{2}(1-\delta)\sum_{j=1}^{n}\int_{Q ^{n}}\left \b\left( -\im\nabla_{\bx_j}+\alpha	\bA^{\rm ex}_{j}\left (\bx_j,Y_{m}\right) \right )\Psi_{n}\right \b^{2} \d X_{n} \nn \\
	&- \frac{1}{2}(\delta^{-1}-1) |\alpha|^2 \sum_{j=1}^{n}\int_{Q ^{n}} \left| \bA^{\rm in}_{j}\left (X_n\right )\right|^2 \left| \Psi_n \right|^2 \d X_{n}
	\end{align}
	by expanding the square and using the Cauchy-Schwarz inequality, where $0 < \delta < 1$ can be chosen freely. The first term in the right-hand side is a a sum of one-body energies in the fermionic wave-function $\Psi_n$. Hence 
	$$ 
	\sum_{j=1}^{n}\int_{Q ^{n}}\left \b\left( -\im\nabla_{\bx_j}+\alpha	\bA^{\rm ex}_{j}\left (\bx_j,Y_{m}\right) \right )\Psi_{n}\right \b^{2} \geq \sum_{k=1} ^n \lambda_k 
	$$
	where $(\lambda_k)_k$ is the sequence of eigenvalues of the one-particle magnetic Laplacian 
	$$
	\left( -\im\nabla_{\bx}+\alpha	\bA^{\rm ex}_{j}\left (\bx,Y_{m}\right) \right )^2.
	$$
	It follows from Lemma~\ref{lem:N} below that if $n$ is large enough, then 
	$$ \sum_{k=1} ^n \lambda_k \geq C n$$
	independently of all parameters of the problem. 
	
	On the other hand, since we work on the unit square, our assumption on $\gamma$ implies a uniform lower bound on $R$, so that, returning to the definition~\eqref{wrr}-\eqref{AJRr}, we have for the second term of right-hand side of~\eqref{eq:smallal}
	$$ |\alpha|^2 \sum_{j=1}^{n}\int_{Q ^{n}} \left| \bA^{\rm in}_{j}\left (X_n\right )\right|^2 \left| \Psi_n \right|^2 \d X_{n} \leq C |\alpha|^2 n ^3.$$
	Inserting in~\eqref{eq:smallal} we find   
	$$ \mathcal{E}^{R}_{n}\left (Q,Y_m\right )\left [\Psi_n\right ] \geq C n (1-\delta) - C' (\delta^{-1}-1)|\alpha|^2 n ^3 $$
	and since we assume that $n<\overline{N}$ for some fixed constant $\overline{N}$, it suffices to choose $|\alpha|$ small enough to deduce that $ \mathcal{E}^{R}_{n}\left (Q,Y_m\right )\left [\Psi_n\right ] $ is bounded below by a positive constant. 
	
	There only remains to combine with the bounds from Step 1 of the proof to obtain, after a last adjustment of constants, the desired conclusion for the whole range of $\alpha$, under the stated constraints on $n$ and $\gamma$. 
	\end{proof}

\subsection{Small boxes}

We now work under the assumption that $\gamma \geq c_{2} > \sqrt{2}$. We shall actually assume more, namely that the particle number is uniformly bounded from above, and that $\gamma$ is sufficiently large compared to that upper bound. We use only the first term in~\eqref{eq:enerinout} to obtain:

\begin{proposition}[\textbf{Exclusion principle on small boxes}]\mbox{}\label{lem:smallexc}\\
	Assume that $\gamma \geq C\sqrt{N}\geq \sqrt{2}$  for a sufficiently large constant $C$. There exists a $\underline{N}$ such that, if $\underline{N} < n \leq N$, then  
	\begin{equation}
		E^{R}_{n}\left ( Q,m \right ) \geq C \frac{n}{\b Q\b}
	\end{equation}
	independently of $\alpha$, $R$ and $m$.
\end{proposition}

A possible proof of the above is as in Step~2 of the proof of Proposition~\ref{lem:medium}, using that $\gamma^{-1}$ (instead of $\alpha$) is large enough to treat the magnetic field generated by particles inside the box perturbatively. We provide another proof, observing that this field is essentially equivalent to an external one in the regime of small boxes.  Indeed, it is constant for $\gamma  > \sqrt{2}$, covering the whole box, and the total flux in the box is small when $\gamma$ is large. One can then reduce to a fermionic one-body problem with an external magnetic field. 

\begin{lemma}[\textbf{Reduction to a one-particle operator}]\label{lem:small1body}\mbox{}\\
Assume $R > \sqrt{2} L$, let $Y_m = (\by_1,\ldots,\by_m) \in (Q^c)^{m}$. Define the vector potential
\begin{equation}\label{eq:Asmall}
\bA (\bx) :=  (n-1)\frac{\bx^\perp}{ R^2}  + \sum_{j= 1} ^m \nabla^{\perp} w_R (\bx-\by_j) 
\end{equation}
associated to the magnetic field 
\begin{equation}\label{eq:Bsmall}
\curl \, \bA (\bx) = \frac{2}{R^2} \left( (n-1) +  \sum_{k=1} ^m \1_{|\bx-\by_k|\leq R}\right).
\end{equation}
Define the Neumann realization of the associated one-particle magnetic Schr\"odinger operator acting on $L^{2}\left (Q\right )$
\begin{equation}\label{eq:hsmall}
h_\bA := \left( -\im \nabla_\bx + \alpha \bA (\bx) \right)^2 = \sum_{k=1} ^\infty \lambda_k |u_k\rangle \langle u_k|  
\end{equation}
together with its spectral decomposition (with eigenvalues labeled in increasing order and $L^2$-normalized eigenfunctions). We have that, for any normalized $\Psi_n \in L_{\rm asym}^2 (Q^n)$ and any $k\in \mathbb{N}$ 
\begin{equation}\label{eq:func small}
\mathcal{E}^{R}_{n}\left (Q,Y_m\right )\left [ \Psi_n\right ] \geq \frac{1}{2} n \lambda_k \left( 1 - \frac{k-1}{n} - C \frac{kn L^2}{R^2}\right).
\end{equation}
\end{lemma}

\begin{proof}
We use only the first term in~\eqref{eq:enerinout} here. By fermionic symmetry 
\begin{multline*} 
2 \mathcal{E}^{R}_{n}\left (Q,Y_m\right )\left [ \Psi_n\right ] \geq \\ n\int_{Q^n} \left|\left( -\im \nabla_{\bx_1} + \alpha \sum_{k= 2}^n \nabla^{\perp} w_R (\bx_1-\bx_k) + \alpha \sum_{j= 1} ^m \nabla^{\perp} w_R (\bx_{1}-\by_j)\right)  \Psi_n \right|^2 \d\bx_1\ldots \d\bx_n.
\end{multline*}
We also use the translation invariance to only consider squares on the form $Q=\left [0,L\right ]^{2}$.
We observe that, under our assumption on $\gamma = R/L$ we have (see the definition \eqref{AJRr}) that $\b\bx_{1}-\bx_{k}\b_{R}=R$ and then that
$$\nabla^{\perp} w_R (\bx_1-\bx_k) = \frac{(\bx_1-\bx_k) ^\perp}{R^2}$$
for $\bx_1,\bx_k \in Q$. Define then  
$$ \varphi (\bx,\by) := - \frac{\bx \cdot \by ^{\perp}}{R^2} \quad \mbox{ and } \quad  \Phi_n := \prod_{k= 2} ^n e^{-\im \alpha\varphi (\bx_1,\bx_k)} \Psi_n .$$
Since 
$$ \nabla_{\bx} \varphi (\bx,\by) = - \frac{\by ^{\perp}}{R^2}$$
we find that 
\begin{multline*}
\int_{Q^n} \left|\left( -\im \nabla_{\bx_1} + \alpha \sum_{k= 2}^n \nabla^{\perp} w_R (\bx_1-\bx_k) + \alpha \sum_{j= 1} ^m \nabla^{\perp} w_R (\bx_{1}-\by_j)\right)  \Psi_n \right|^2 \d\bx_1\ldots \d\bx_n = \\ 
\int_{Q^n} \left|\left( -\im \nabla_{\bx_1} + \alpha (n-1) \frac{\bx_1^\perp}{ R^2} + \alpha \sum_{j= 1} ^m \nabla^{\perp} w_R (\bx_{1}-\by_j)\right)  \Phi_n \right|^2 \d\bx_1\ldots \d\bx_n = \\
\int_{Q^n} \left|\left( -\im \nabla_{\bx_1} + \alpha \bA (\bx_1) \right)  \Phi_n \right|^2 \d\bx_1\ldots \d\bx_n.
\end{multline*}
Inserting now the spectral decomposition of $h_\bA$
\begin{align*}
\int_{Q^n} \left|\left( -\im \nabla_{\bx_1} + \alpha \bA (\bx_1) \right)  \Phi_n \right|^2 \d\bx_1\ldots \d\bx_n &= \sum_{j=1} ^\infty \lambda_j \int_{Q^{n-1}} \left| \int_{Q} \overline{u_j (\bx_1)} \Phi_n (\bx_1,\ldots,\bx_n) \d\bx_1 \right|^2 \d\bx_2 \ldots \d\bx_n \\
&\geq \lambda_k \left( 1 - \sum_{j=1} ^{k-1} \int_{Q^{n-1}} \left| \int_{Q} \overline{u_j (\bx_1)} \Phi_n (\bx_1,\ldots,\bx_n) \d\bx_1 \right|^2 \d\bx_2 \ldots \d\bx_n \right)
\end{align*}
where we picked an energy level $\lambda_k$ and dropped the contribution of all levels with $j< k$. Next we split  
\begin{multline*}
\int_{Q^{n-1}} \left| \int_{Q} \overline{u_j (\bx_1)} \Phi_n (\bx_1,\ldots,\bx_n) \d\bx_1 \right|^2 \d\bx_2 \ldots \d\bx_n \\= \int_{Q^{n}} \overline{u_j (\bx_1)} u_j (\by_1)\Phi_n (\bx_1,\ldots,\bx_n) \overline{\Phi_n (\by_1,\ldots,\bx_n)} \d \by_1 \d\bx_1 \d\bx_2 \ldots \d\bx_n = \mathrm{I} + \mathrm{II}
\end{multline*}
with 
\begin{align*} 
\mathrm{I}\,&= \int_{Q^{n}} \overline{u_j (\bx_1)} u_j (\by_1)\Psi_n (\bx_1,\ldots,\bx_n) \overline{\Psi_n (\by_1,\ldots,\bx_n)} \d \by_1 \d\bx_1 \d\bx_2 \ldots \d\bx_n\\
\mathrm{II}\,&= \int_{Q^{n}} \left( \prod_{k= 2} ^n e^{\im\alpha \varphi(\bx_1,\bx_k) - \im\alpha \varphi(\by_1,\bx_k)} -1\right)\overline{u_j (\bx_1)} u_j (\by_1)\Psi_n (\bx_1,\ldots,\bx_n) \overline{\Psi_n (\by_1,\ldots,\bx_n)} \d \by_1 \d\bx_1 \d\bx_2 \ldots \d\bx_n.
\end{align*}
Now, the first term $\mathrm{I}$ is just $n^{-1}$ times the occupation number of the mode $u_j$ in the \emph{fermionic} wave-function $\Psi_n$. Hence 
$$ |\,\mathrm{I}\,| \leq \frac{1}{n}.$$
On the other hand, since all variables sit within the box $Q=\left [0,L\right ]^{2}$ we clearly have 
$$ \left| \varphi(\bx_1,\bx_k) - \varphi(\by_1,\bx_k)\right| \leq C \frac{L^{2}}{R^2}$$
 and hence 
$$ \left|1- \prod_{k= 2} ^n e^{\im\alpha \varphi(\bx_1,\bx_k) - \im\alpha \varphi(\by_1,\bx_k)} \right| \leq C \b \alpha\b\frac{nL^2}{R^2}.$$
Thus, with the Cauchy-Schwartz inequality
\begin{align*}
|\mathrm{II}| &\leq C \frac{nL^2}{R^2} \int_{Q^{n}} |u_j (\bx_1)| |u_j (\by_1)| \left|\Psi_n (\bx_1,\ldots,\bx_n) \right| \left|\Psi_n (\by_1,\ldots,\bx_n)\right| \d \by_1 \d\bx_1 \d\bx_2 \ldots \d\bx_n\\
&\leq C \frac{nL^2}{R^2} \int_{Q^{n}} |u_j (\bx_1)|^2 \left|\Psi_n (\by_1,\ldots,\bx_n)\right| ^2 \d \by_1 \d\bx_1 \d\bx_2 \ldots \d\bx_n = C \frac{nL^2}{R^2}. 
\end{align*}
The desired final result is obtained by collecting the previous inequalities.
%
\end{proof}

Now we need lower bounds on the $\lambda_{k}$'s, eigenvalues of a fermionic one-body problem with an external magnetic field. The latter can be quite general, because of the influence of particles sitting outside the box. Diamagnetic considerations however imply estimates independent of this field, see Appendix~\ref{sec:app}. Based on Lemma~\ref{lem:N} below we can conclude the 

\begin{proof}[Proof of Proposition~\ref{lem:smallexc}]
Scaling length units by a factor $L$, we apply Lemma~\ref{lem:N} to the magnetic Schr\"odinger operator~\eqref{eq:hsmall}. We thus find that the number of energy levels below a threshold $\Lambda L^{-2}$ is bounded above by a function of $\Lambda$ only. Pick some value of $\Lambda$, say $\Lambda = 2$ for reference and denote $\underline{N}-1$ the number of eigenvalues of $h_A$ below $\Lambda L^{-2}$. Using~\eqref{eq:func small} (with a proper choice of $\lambda_k$ in relation with $2 L^{-2}$) and inserting the consequences of~Lemma~\ref{lem:N} we just mentioned, we find 
$$
\mathcal{E}^{R}_{n}\left (Q,Y_m\right )\left [ \Psi_n\right ] \geq C L^{-2} n \left( 1 - \left (\underline{N}-1\right ) n^{-1} - C \underline{N} n L^2 R^{-2}\right).
$$
Under our stated assumptions on $n$ and $\gamma$, the quantity inside the parenthesis is bounded below by a positive universal constant, and this concludes the proof.
\end{proof}

\subsection{Conclusion of proofs}

First note that Proposition~\ref{lem:exclbox} follows immediately from the combination of Lemma~\ref{lem:finite}, Propositions~\ref{lem:smallexc}-\ref{lem:medium} and \ref{lem:excllarge}:

\begin{proof}[Proof of Proposition~\ref{lem:exclbox}]
Let $\underline{N}$ be the minimum of the lower bounds on $n$ appearing in Propositions~\ref{lem:smallexc} and \ref{lem:medium}. Let $\overline{N}=4^k$ be such that 
$$ \underline{N} \leq 4^{k-1} + 1.$$
Using~\eqref{eq:lowest} and applying Lemma~\ref{lem:finite} we find for $N\geq \overline{N}$
$$
E^{R}_{N} \left (Q,m\right ) \geq E (N,Q) \geq C_{k} N \min\left(E(4^{k-1}+1,Q), \ldots, E(4^k,Q)\right). 
$$
Note that the assumption $E(n,Q) >0$ for all $ 4^{k-1}+1 \leq n\leq 4N$ made in Lemma~\ref{lem:finite} is valid as shown by a combination of Propositions~\ref{lem:smallexc}-\ref{lem:medium} and~\ref{lem:excllarge} and Lemma~\ref{lem:ap_E2}.

We are thus reduced to a uniform strictly positive lower bound on $|Q| E(n)$ for particle numbers $n$ with $ \underline{N} \leq n \leq \overline{N}$. We obtain this using bounds on $E (n,Q)$ provided by Propositions~\ref{lem:excllarge}-\ref{lem:medium} or~\ref{lem:smallexc} depending on the value of $\gamma = R/L$. If $\gamma \geq C (\overline{N}) ^{1/2}$ for $C$ sufficiently large, we may use Proposition~\ref{lem:smallexc}. We use~\eqref{eq:gage1} if $\gamma < c_1$, the constant in the statement of Proposition~\ref{lem:excllarge}. Finally we use Proposition~\ref{lem:medium} in the remaining range of $\gamma$. In all cases we find 
$$ E (n,Q) = \inf_{R,m} E^{R}_{n}\left(Q,m\right) \geq \frac{C|1-\alpha|}{|Q|}$$
when $\underline{N} \leq n \leq \overline{N}$, independently of $R$, and this yields the result.
\end{proof}

We next use Proposition~\ref{lem:exclbox} to conclude the 

\begin{proof}[Proof of Theorem~\ref{LEP_an}]
We  work under the assumption that $\Psi_{N}$ has a density satisfying
\begin{equation}\label{as:rho3}
	1 + \underline{N} \leq  N_< \leq \int_{Q}\rho_{\Psi_{N}}\left (\bx\right )\d \bx\leq N_>
\end{equation}
where $\underline{N}$ is again the minimum of the lower bounds on $n$ appearing in Propositions~\ref{lem:smallexc} and \ref{lem:medium}. We split the quantity to estimate according to how many particles are in the square $Q$
\begin{equation}\label{eq:split}
	\mathcal{E}_{Q}^{R}\left [\Psi_{N}\right ] \geq \sum_{n=0}^{N}	E_{n}^{R}\left (Q,m\right )p_{n}\left (\Psi_{N},Q\right )
\end{equation}
with $p_{n}\left (\Psi_{N},Q\right )$ the $n$-particle probability distribution induced from $\Psi_{N}$
\begin{equation}
	p_{n}\left (\Psi_{N},Q\right )=\sum_{A\subseteq \left \{1,\cdots ,N\right \},\b A\b =n}\int_{(Q^{c})^{N-n}}\int_{Q^{n}}\left \b\Psi_{N}\right \b^{2}\prod_{k\in A}\d \bx_{k}\prod_{l\notin A}\d \bx_{l}\nn
\end{equation}
satisfying
\begin{equation}\label{eq:prob}
	\sum_{n=0}^{N}	p_{n}\left (\Psi_{N},Q\right )=1\;\;\text{and}\;\;	\sum_{n=0}^{N}	np_{n}\left (\Psi_{N},Q\right )=\int_{Q}\rho_{\Psi_{N}}.
\end{equation}
 Then, using Proposition~\ref{lem:exclbox} we obtain  
\begin{equation}\label{eq:split2}
	\mathcal{E}_{Q}^{R}\left [\Psi_{N}\right ] \geq \frac{C\b \alpha -1\b}{|Q|} \sum_{n=\underline{N}}^{N} n p_{n}\left (\Psi_{N},Q\right ).\nn
\end{equation}
On the other hand, using~\eqref{as:rho3} and~\eqref{eq:prob} we have 
\begin{align}\label{eq:sufficient}
\sum_{n\geq \underline{N}} n p_{n}\left (\Psi_{N},Q\right ) &\geq \underline{N} + 1 - \sum_{n < \underline{N}} n p_{n}\left (\Psi_{N},Q\right )  \nonumber\\
&\geq \underline{N} + 1 - (\underline{N} - 1)\sum_{n\geq 0} p_{n}\left (\Psi_{N},Q\right)\nonumber \\
&\geq 2\nn
\end{align}
and hence 
$$\mathcal{E}_{Q}^{R}\left [\Psi_{N}\right ] \geq \frac{C\b \alpha -1\b}{|Q|N_>} \int_{Q}\rho_{\Psi_{N}}\left (\bx\right )\d \bx,$$
concluding the proof.
\end{proof}

\appendix

\section{Diamagnetic estimates with Neumann boundary conditions}\label{sec:app}

Here we state the lemma on eigenvalues of one-body magnetic Laplacians we have used twice in the paper. As already mentioned in Remark~\ref{rem:fermions}, it can be obtained following general arguments in~\cite{Frank-09,HunSim-04}. We provide our self-contained proof for the convenience of the reader.

\begin{lemma}[\textbf{Diamagnetic bound}]\mbox{}\label{lem:N}\\
	Consider the Neumann realization of the magnetic Laplacian
\begin{equation}
	H_{\bA}:=\left (-\im\nabla +\bA\right )^{2}\nn
\end{equation}
on $H^{1}_{\bA}\left (Q_{0}\right ),$ where $\bA$ and $\curl \bA$ are bounded functions. Define the number
	\begin{equation}
		N\left (\Lambda , \bA\right ):=\text{the number of eigenvalues of}\;H^{\bA}\;\text{less than or equal to}\; \Lambda.\nn
	\end{equation}
	There exists a function $f:\R\to \R$ independent of $\bA$ such that 
	\begin{equation}
		N\left (\Lambda , \bA\right )\leq f\left (\Lambda\right ).\nn
	\end{equation}
\end{lemma} 

We first reduce to the case where $\vec{\nu} \cdot \bA \equiv 0$ on the boundary. 

\begin{lemma}[\textbf{Reduction to tangential vector potentials}]\label{lem:reducA}\mbox{}\\
The general case of Lemma~\ref{lem:N} is implied by the particular case where the normal component of $\bA$ vanishes on the boundary, $\vec{\nu} \cdot \bA \equiv 0$ on $\partial Q$.  
\end{lemma}

\begin{proof}
Let $\phi$ be the unique solution to the Dirichlet problem 
$$ 
\begin{cases}
 \Delta \phi = B = \curl (\bA) \mbox{ in } Q_0\\
 \phi = 0 \mbox{ on } \partial Q_0.
\end{cases}
$$
Note that  $\phi$ is well defined and that $\nabla\phi$ has a well-defined trace on $\partial Q$ since $B \in L^{\infty}$. This follows from elliptic
regularity theory.
By definition 
$$ \curl \bA = \curl \nabla^{\perp} \phi$$
and hence there exists some $\varphi$ such that 
$$ \bA = \nabla^{\perp} \phi + \nabla \varphi \mbox{ in } Q_0.$$
Thus we may change gauge
$$ e^{\im \varphi} H_{\bA} e^{-\im \varphi} = H_{\nabla^{\perp} \phi}$$
where the multiplication operator $e^{\im \varphi}$ is unitary. Hence $H_{\bA}$ and $H_{\nabla^{\perp} \phi}$ have the same spectrum. Since $\phi$ is constant on the boundary, clearly $\vec{\nu} \cdot \nabla^{\perp} \phi \equiv 0$ there.
\end{proof}

From now on we thus assume that
\begin{equation}
	\left \{
	\begin{array}{r c l}
		\curl \bA &=B  \;\;\;\;\;\;\;\;\;\;\;\;\;\;\\
		\bA\cdot \vec{\nu} &= 0 \;\;\text{on} \;\; \partial Q_{0}\nn
	\end{array}
	\right .  
\end{equation}
where $\vec{\nu}$ is the normal vector of $\partial Q_{0}$.
We denote $\left \{\lambda_{j}\left (\bA\right )\right \}_{j=1}^{\infty}$ the eigenvalues of $H_{\bA}$. The associated eigenfunctions $\Psi_j$ are solutions of
\begin{equation}
	\left \{
	\begin{array}{r c l}
		\left (-\im\nabla +\bA\right )^{2}\Psi_{j}=&\lambda_{j}\Psi_{j} \;\;\text{on} \;\; Q_{0} \\
		\vec{\nu}   \cdot\left (-\im\nabla +\bA\right )\Psi_{j}=&\left (-\im\nabla\cdot \vec{\nu}\right )\Psi_{j} =0 \;\;\text{on} \;\; \partial Q_{0}.\nn
	\end{array}
	\right .  
\end{equation}
For any real number $e >0$
we define the magnetic Neumann Green function $G^{\bA , e}_{\by}$ to be the solution of
\begin{equation}
	\left \{
	\begin{array}{r c l}
		\left (H_{\bA}+e\right )G^{\bA , e}_{\by}\left (\bx\right )=&\delta_{\by}\left (\bx\right ) \;\;\text{on} \;\; Q_{0} \\
		\left (-\im\nabla\cdot \vec{\nu}\right )G^{\bA , e}_{\by}\left (\bx\right ) =&0 \;\;\text{on} \;\; \partial Q_{0}.\nn
	\end{array}
	\right .  
	\label{def_green}
\end{equation}
in the weak sense that
\begin{equation}
	\int_{Q}G^{\bA , e}_{\by}\left (\bx\right ) \left (-\im\nabla +\bA\right )^{2}\phi\left (\bx\right )\d \bx +e\int_{Q}G^{\bA , e}_{\by}\left (\bx\right )\phi\left (\bx\right )\d \bx =\phi\left (\by\right )
	\label{weakgreen}\nn
\end{equation}
for all $\phi \in C^{\infty}\left (Q\right )$. 

This way, the function
\begin{equation}
	u_{g}\left (\bx\right )=\int_{Q}G^{\bA , e}_{\by}\left (\bx\right )g\left (\by\right )\d \by
	\label{ug}\nn
\end{equation}
$g\in C^{\infty}\left (Q\right )$ is the unique solution of 
\begin{equation}
	\left (H_{\bA}+e\right )u_{g}\left (\bx\right )=g\left (\bx\right )
	\label{weaku}\nn
\end{equation}
in the above weak sense (cf the Lax-Milgram theorem). In other words $G^{\bA , e}_{\by} (\bx)$ is the integral kernel of $\left (H_{\bA}+e\right )^{-1}.$ 

Our aim is now to show that the Green function with magnetic field is always smaller than the one without magnetic field. To this end we use Kato's inequality as in~\cite[Section~4.4]{LieSei-09}. We need a version thereof valid in the case of Neumann boundary conditions:
\begin{lemma}[\textbf{Kato's Inequality with Neumann boundary conditions}]\mbox{}\label{kato_ine_neum}\\
	Let $u\in H^{1}\left (Q\right )$ such that  $ \left (-\im\nabla\cdot \vec{\nu}\right )u =0$ and $H_{\bA}u$ be in $L^{1}_{\mathrm{loc}}\left (Q\right )$. Let $\mathrm{sgn}\left (u\right )=\bar{u}/\left \b u\right \b$ if $u\left (\bx\right )\neq 0$ and $0$ otherwise. Then
	\begin{equation}
		-\Delta\left \b u\right \b\leqslant \mathrm{Re}\left [\mathrm{sgn}\left (u\right )H_{\bA}u\right ]\nn
	\end{equation}
	in the weak sense that the inequality holds when integrated against any non-negative $\phi \in C^{\infty}\left (Q\right )$ with $\vec{\nu} \cdot \nabla \phi \equiv 0$ on the boundary, i.e. 
	$$ - \int_{Q} |u| \Delta \phi \leq \int_{Q} \mathrm{Re}\left [\mathrm{sgn}\left (u\right )H_{\bA}u\right ] \phi $$
\end{lemma}

\begin{proof}
	We may follow the proof of~\cite[Theorem- X.33]{ReeSim2}. Essentially, if $u$ is smooth and non-vanishing, the result holds pointwise. Hence the inequality holds pointwise for smooth non-vanishing functions. Using a standard regularisation we deduce that the inequality holds in the weak sense by performing integrations by parts. The boundary terms vanish because we have~$\vec{\nu} \cdot \nabla \phi \equiv 0$ and $\vec{\nu} \cdot \nabla |u| \equiv 0$ on the boundary. 
	
	A complete proof is obtained as in~\cite[Theorems~X.27~and~X.33]{ReeSim2} by working with the regularized absolute value
	$$
		u_{\epsilon}\left (\bx\right )=\sqrt{\b u\left (\bx\right )\b ^2 +\epsilon^{2}}.
	$$
	The additional ingredient is the observation that $\vec{\nu} \cdot \nabla u \equiv 0$ on the boundary implies that also $\vec{\nu} \cdot \nabla u_{\epsilon} \equiv 0$ on the boundary, so that integration by parts do not produce boundary terms.
	
%
\end{proof}

We prove another intermediary lemma.
\begin{lemma}[\textbf{Positivity of the Laplace Green function}]\mbox{}\label{lem:posGe}\\
	The operator
	\begin{align}
		\left (-\Delta +e\right )^{-1}: L^{2}\left (Q\right )&\to  H^{1}\left (Q\right )\nn\\
		f&\mapsto u\nn
	\end{align} 
	defined by 
	\begin{equation}
		\left \{
		\begin{array}{r c l}
			-\Delta u + eu=&f \;\;\text{on} \;\; Q \nn\\
			\vec{\nu}   \cdot\nabla u=&0 \;\;\text{on} \;\; \partial Q.\nn
		\end{array}
		\right .  
	\end{equation}
	preserves positivity, i.e, $	\left (-\Delta +e\right )^{-1}f\geq 0$ if $f \geq 0$. Hence the Green function $G^{0,e}$ with zero magnetic field is non-negative.
\end{lemma}
\begin{proof}
	The function $u$ being the solution of $	-\Delta u + eu=f$ we know by the Lax-Milgram theorem that $u$ is the unique minimizer on $H^{1}\left (Q\right )$ of the energy\begin{equation}
		\mathcal{E}\left [u\right ]=\frac{1}{2}\left (\int \left \b \nabla u\right \b^{2}+e\b u\b^{2}\right )-\int fu\nn
	\end{equation}
	but since $f\geq 0$ we clearly have $\mathcal{E}\left [u\right ]\geq \mathcal{E}\left [\b u\b\right ]$ and thus $u=\b u\b\geq 0$ by uniqueness of the solution.
\end{proof}

We have all the ingredients to compare the Green functions with and without magnetic field:

\begin{lemma}[\textbf{Diamagnetic inequality for Green functions}]\mbox{}\label{lem:ineG}\\
	Let $ G^{\bA , e}_{\bx}$ and $ G^{ e}_{\bx}=G^{0 , e}_{\bx}$ be defined as in \eqref{def_green}. For almost every $\bz \in Q$
	\begin{equation}
		\left \b G^{\bA , e}_{\bx}\left (\bz\right )\right \b \leqslant G^{e}_{\bx}\left (\bz\right )\nn
	\end{equation}
\end{lemma}

\begin{proof}
	We take two positive functions $f,h\in C^{\infty}\left (Q\right )$ and define
	\begin{equation}
		u_{h}\left (\bx\right )=\int_{Q}G^{\bA , e}_{\by}\left (\bx\right )h\left (\by\right )\d \by\;\;\text{and}\;\;u_{f}\left (\bx\right )=\int_{Q}G^{e}_{\by}\left (\bx\right )f\left (\by\right )\d \by
		\label{uhuf}.\nn
	\end{equation}
	We know by Lemma \ref{lem:posGe} that $u_{f}$ is positive. We apply Lemma \ref{kato_ine_neum} to $u_{h}$ to obtain 
	\begin{equation}
		\left (-\Delta +e\right ) \left \b u_{h}\right \b\left (\bx\right )\leqslant h\left (\bx\right )
		\label{inesur_u}.\nn
	\end{equation}
	We now multiply by $u_{f}$ and integrate to get that
	\begin{align}
		\int_{Q^{2}}G^{e}_{\by}\left (\bx\right )f\left (\by\right )\left (-\Delta_{\bx} +e\right )\left \b u_{h}\left (\bx\right )\right \b \d \bx\d \by\leqslant \int_{Q^{2}}G^{e}_{\by}\left (\bx\right )f\left (\by\right )h\left (\bx\right )\d \bx\d \by\nn\\
		\int_{Q}f\left (\by\right )\left \b u_{h}\left (\by\right )\right \b \d \by\leqslant \int_{Q^{2}}G^{e}_{\by}\left (\bx\right )f\left (\by\right )h\left (\bx\right )\d \bx\d \by\nn
	\end{align}
	where we used that $\left (-\Delta_{\bx} +e\right )G^{e}_{\by}\left (\bx\right )=\delta_{\by}\left (\bx\right )$. We then obtain (recalling that $f\geq 0$),
	\begin{align}
		\left \b\int_{Q^{2}}G^{\bA , e}_{\bx}\left (\by\right )f\left (\by\right )h\left (\bx\right )\d \bx\d \by\right \b &\leqslant \int_{Q^{2}}G^{e}_{\by}\left (\bx\right )f\left (\by\right )h\left (\bx\right )\d \bx\d \by\nn\\
		\left \b G^{\bA , e}_{\bx}\left (\bz\right )\right \b &\leqslant G^{e}_{\bx}\left (\bz\right ) \nn
	\end{align}
	by taking $h\to\delta_{\bx}$ and $f\to\delta_{\by}$.
\end{proof}

Now we may conclude the 

\begin{proof}[Proof of Lemma~\ref{lem:N}]
	We define for two given numbers $e>0$ and $\Lambda>0$, the Birman-Schwinger operator
	\begin{equation}
		K_{e}:=\sqrt{\Lambda +e}\left (H^{\bA}+e\right )^{-1}\sqrt{\Lambda+e}\nn
	\end{equation}
	This way if $H^{\bA}\psi =\lambda\psi$ we have
	$$K_{e} \psi=\frac{\Lambda +e}{\lambda +e} \psi. $$
 Hence eigenvalues of $H^{\bA}$ with $\lambda\leq \Lambda$ correspond to eigenvalues of $K_{e}$ larger than one. So, if we denote
	\begin{equation}
		B_{e}:=\text{the number of eigenvalues of}\;K_{e}\;\text{larger than or equal to 1},\nn
	\end{equation}
	the Birman-Schwinger principle (see \cite[Equation 4.3.5]{LieSei-09}) states that
	\begin{equation}
		N\left (\Lambda , \bA\right )=B_{e}.
	\end{equation}
Since we work with a bounded magnetic field and magnetic vector potential, $\left (H^{\bA}\right )^{-1}:L^{2}\left (Q_{0}\right )\to H^{1}\left (Q_{0}\right )$ with $H^{1}\left (Q_{0}\right )$ compactly embeded in $L^{2}\left (Q_{0}\right )$. $K_{e}$ is consequently compact, but also Hilbert-Schmidt because $H^{\bA}$ is a bounded perturbation of the 2D Laplacian.
	We bound $B_{e}$ in the following way, for any $m\geqslant 1$
	\begin{align}
		B_{e}&\leqslant \Tr\left [K_{e}^{m}\right ]\nn\\
		&=\left (\Lambda+e\right )^{m}\Tr \left [\left (H^{\bA}+e\right )^{-m}\right ]\nn\\
		&=\left (\Lambda+e\right )^{m}\int_{Q^{2m}}G_{e}^{\bA}\left (\bx, \mathbf{y}_{1}\right )G_{e}^{\bA}\left (\mathbf{y}_{1}, \mathbf{y}_{2}\right )...G_{e}^{\bA}\left (\mathbf{y}_{m-1}, \bx\right )\d \bx \d \mathbf{y}_{1} \d \mathbf{y}_{2}... \d \mathbf{y}_{m-1}\nn
	\end{align}
	we take the absolute value and use Lemma \eqref{lem:ineG} to get
	\begin{align}
		B_{e}&\leqslant \left (\Lambda+e\right )^{m}\int_{Q^{2m}}G_{e}^{0}\left (\bx, \mathbf{y}_{1}\right )G_{e}^{0}\left (\mathbf{y}_{1}, \mathbf{y}_{2}\right )...G_{e}^{0}\left (\mathbf{y}_{m-1}, \bx\right )\d \bx \d \mathbf{y}_{1} \d \mathbf{y}_{2}... \d \mathbf{y}_{m-1}\nonumber\\
		&=\left (\Lambda+e\right )^{m}\Tr\left (-\Delta_{\bx}+e\right )^{-m}\nonumber\\
		&=\left (\Lambda+e\right )^{m}\sum_{j=1}^{\infty}\left (\lambda_{j}\left (0\right )+e\right )^{-m}\nonumber\\
		&=:f\left (\Lambda\right )
	\end{align}
	where $f$ is finite for $m\geq 2$ and independent of $\bA$.
\end{proof}
\bigskip\bigskip\bigskip\bigskip\bigskip\bigskip\bigskip\bigskip\bigskip\bigskip\bigskip\bigskip\bigskip\bigskip\bigskip\bigskip\bigskip\bigskip\bigskip\bigskip\bigskip\bigskip\bigskip\bigskip\bigskip\bigskip\bigskip\bigskip\bigskip\bigskip\bigskip\bigskip\bigskip\bigskip\bigskip\bigskip\bigskip\bigskip\bigskip\bigskip\bigskip\bigskip\bigskip\bigskip\bigskip\bigskip\bigskip
On behalf of all authors, Théotime Girardot states that there is no conflict of interest.

%
\end{document}